\documentclass [fleqn,10pt] {article}
\usepackage{amsfonts}
\usepackage{amsmath}
\usepackage{amssymb}
\usepackage{hyperref}
\usepackage{bm}
\usepackage[alphabetic]{amsrefs}
\usepackage{fullpage,theorem}
\usepackage{tensor}
\usepackage{cancel}
\usepackage[greek,english]{babel}

\newtheorem{thm}{Theorem}[section]
\newtheorem{cor}[thm]{Corollary}
\newtheorem{lem}[thm]{Lemma}
\newtheorem{prop}[thm]{Proposition}
\newtheorem{prop-defn}[thm]{Proposition/Definition}

\theorembodyfont{\rmfamily}
\newtheorem{defn}[thm]{Definition}

\newtheorem{rem}[thm]{Remark}

\numberwithin{equation}{section}
\newenvironment{proof}{\noindent \emph{Proof.}}{\hspace{\stretch{1}}$\Box$}

\newcommand{\parderv}[2] {\frac{\partial#1}{\partial#2}}

\newcommand{\mcK} {\mathcal{K}}
\newcommand{\mcL} {\mathcal{L}}
\newcommand{\mcM} {\mathcal{M}}
\newcommand{\mcN} {\mathcal{N}}

\newcommand{\mcS} {\mathcal{S}}

\newcommand{\mcU} {\mathcal{U}}

\newcommand{\dd} {\mathrm{d}}

\newcommand{\ii} {\mathrm{i}}

\newcommand{\ind} {\indices}

\newcommand{\lb} [1] {{\left[ #1 \right. }}
\newcommand{\rb} [1] {{\left. #1 \right] }}

\newcommand{\Rho} {\mathrm{P}}

\newcommand{\SL} {\mathrm{SL}}

\newcommand{\SO} {\mathrm{SO}}

\newcommand{\SU} {\mathrm{SU}}

\newcommand{\Tgt} {\mathrm{T}}

\newcommand{\CP} {\mathbb{CP}}

\newcommand{\R} {\mathbb{R}}
\newcommand{\C} {\mathbb{C}}

\newcounter{mnotecount}[section]
\renewcommand{\themnotecount}{\thesection.\arabic{mnotecount}}

\newcommand{\mnote}[1]
{\protect{\stepcounter{mnotecount}}$^{\mbox{\footnotesize
$
\bullet$\themnotecount}}$ \marginpar{
\raggedright\tiny\em
$\!\!\!\!\!\!\,\bullet$\themnotecount: #1} }

\begin{document}
\title{A Goldberg-Sachs theorem in dimension three}
\author{Pawe\l~ Nurowski\footnote{Centrum Fizyki Teoretycznej PAN 
Al. Lotnik\'{o}w 32/46 
02-668 Warszawa, Poland. Email: \texttt{nurowski@cft.edu.pl} } 
\, \& Arman Taghavi-Chabert\footnote{{Masaryk University, Faculty of Science, Department of Mathematics and Statistics, Kotl\'{a}\v{r}sk\'{a} 2, 611 37 Brno, Czech Republic. Email: \texttt{taghavia@math.muni.cz}}
}}
\date{}

\maketitle

\begin{abstract}
We prove a Goldberg-Sachs theorem in dimension three. To be precise, given a three-dimensional Lorentzian manifold satisfying the topological massive gravity equations, we provide necessary and sufficient conditions on the tracefree Ricci tensor for the existence of a null line distribution whose orthogonal complement is integrable and totally geodetic. This includes, in particular, Kundt spacetimes that are solutions of the topological massive gravity equations.
\end{abstract}

\section{Introduction}
The classical Goldberg-Sachs theorem \cite{Goldberg2009} states that \emph{a four-dimensional Ricci-flat Lorentzian manifold $(\mcM,\bm{g})$ admits a shear-free congruence of null geodesics if and only if its  Weyl tensor is algebraically special}. Here, the property of being algebraically special is based on the Petrov classification of the Weyl tensor \cite{Petrov2000}. Since its original publication in 1962 the theorem has been generalised along two main directions:
\begin{itemize}
\item First, it was shown to hold for \emph{Einstein} Lorentzian manifolds, i.e. $R_{ab}=\Lambda \, g_{ab}$ for some constant $\Lambda$, and more general energy momentum tensors, such as $R_{ab}=\Lambda \, g_{ab}+\Phi \,k_a k_b$, for some function $\Phi$ and $1$-form $k_a$ with $k_a k_b g^{ab}=0$. Even weaker conditions involving the Cotton tensor
have been formulated  \cites{Kundt1962,Robinson1963} for which the theorem holds too, thereby highlighting its \emph{conformal invariance}.

\item Second, the theorem admits versions in any metric signatures  \cites{Przanowski1983,Nurowski1993,Nurowski1996,Apostolov1997,Nurowski2002,Gover2011}, providing, among others, an interesting result in four-dimensional Riemannian geometry stating that \emph{a four-dimensional Einstein Riemannian manifold locally admits a Hermitian structure if and only if its Weyl tensor is algebraically special}. The key to the understanding of this generalisation is the fact that in four dimensions, a shear-free congruence of null geodesics and a Hermitian structure are both equivalent to an \emph{integrable totally null complex $2$-plane distribution}. The distinction between them is made by different reality structures. Thus, the Goldberg-Sachs theorem relates the existence of integrable null $2$-plane distributions to the algebraic speciality of the  Weyl tensor.
\end{itemize}

The recent interest in solutions of Einstein equations in higher dimensions has generated much research into the generalisation of the Petrov classification of the Weyl tensor and the Goldberg-Sachs theorem to  higher dimensions. One approach to the problem, advocated by \cites{Hughston1988,Nurowski2002}, is to consider an \emph{(almost) null structure}, i.e. a totally null complex $m$-plane distribution, on a $2m$-dimensional (pseudo-)Riemannian manifold. Their importance in higher-dimensional black holes was highlighted in \cite{Mason2010}. Motivated by the conformal invariance and the underlying complex geometry of the theorem in dimension four, one of the authors (AT-C) proved a Goldberg-Sachs theorem in dimension five in \cite{Taghavi-Chabert2011} and higher in \cite{Taghavi-Chabert2012}. Particularly relevant here is the version of the theorem in dimension $2m+1$, in which an almost null structure is also defined to be a totally null complex $m$-plane distribution $\mcN$, say. The difference now is that $\mcN$ has an orthogonal complement $\mcN^\perp$ of rank $m+1$, and the crucial point, here, is that the theorem of \cite{Taghavi-Chabert2012} gives sufficient, conformally invariant, conditions on the Weyl tensor and the Cotton tensor for the integrability of \emph{both $\mcN$ and $\mcN^\perp$}, i.e.
$[ \Gamma ( \mcN ) , \Gamma ( \mcN ) ] \subset \Gamma ( \mcN )$ and $[ \Gamma ( \mcN^\perp ) , \Gamma ( \mcN^\perp ) ] \subset \Gamma ( \mcN^\perp )$. More refined algebraic classifications of the curvature tensors depending on the concept of almost null structure can be found in \cites{Taghavi-Chabert2012a,Taghavi-Chabert2013,Taghavi-Chabert2014} inspired by \cite{Jeffryes1995}. An alternative approach to the classification of the Weyl tensor of higher-dimensional Lorentzian manifolds is given in \cite{Coley2004}, and a Goldberg-Sachs theorem in this setting has been given in \cites{Durkee2009,Ortaggio2012a,Ortaggio2013a}.

The aim of this paper is to consider yet another generalisation, by formulating the theorem in \emph{three dimensions}. A priori, one would expect such a putative Goldberg-Sachs theorem to follow the same lines as in four and higher dimensions. However, there are a number of features of (pseudo-)Riemannian geometry specific to dimension three that prevent such a straightforward generalisation.
\begin{itemize}
\item First, there is no Weyl tensor.
\item Then, Einstein metrics are necessarily of constant curvature.
\item Finally, an almost null structure (i.e. a totally null (complex) line distribution) $\mcN$ is always integrable, and the (conformally invariant) condition that its orthogonal complement $\mcN^\perp$ be integrable too does not impose any constraint on the curvature.  Related to this is the fact that congruences of null geodesics are necessarily shear-free.
\end{itemize}
To remedy these shortcomings, one is led to seek stronger conditions that must depend on the \emph{metric} rather than the conformal structure of our manifold. To this end, we shall exploit the following special features of three-dimensional (pseudo-)Riemannian geometry.
\begin{itemize}
\item From an algebraic point of view, the tracefree Ricci tensor $\Phi_{ab}$ in dimension three behaves in the same way as the (anti-)self-dual Weyl tensor in dimension four, since they both belong to a five-dimensional irreducible (complex) representation of $\SL(2,\C)$ or any of its real forms. This leads to a notion of \emph{algebraically special} tracefree Ricci tensors in dimension three analogous to the one on the (anti-)self-dual Weyl tensor in dimension four, and a notion of \emph{multiple principal null structure}, i.e. a preferred (complex) null line distribution.

\item The Cotton tensor $A_{abc}$ can be Hodge-dualised to yield a tracefree symmetric tensor $(*A)_{ab}$ of valence $2$, and must then belong to the same (complex) representation as $\Phi_{ab}$.

\item These properties allow us to weaken the Einstein equations to the equations governing \emph{topological massive gravity}, \cites{Deser1982}, which relate $\Phi_{ab}$ and $(*A)_{ab}$ as
\begin{align*}
\Phi_{ab} & = \frac{1}{m} (*A)_{ab} \, , &
R & = 6 \, \Lambda = \mbox{constant} \, ,
\end{align*}
where $\Lambda$ is the cosmological constant and $m$ is a `mass' parameter. In analogy to the Einstein equations, these equations can also be written as
\begin{align*}
R_{ab} - \frac{1}{2} g_{ab} R + \Lambda \, g_{ab} - \frac{1}{m} (*A)_{ab} & = 0 \, .
\end{align*}

\item There is a natural \emph{non-conformally invariant} condition that a (multiple principal) null structure $\mcN$ can satisfy, namely that not only $\mcN$ and $\mcN^\perp$ be integrable, but that $\mcN^\perp$ be \emph{totally geodetic}, i.e. $\bm{g} ( \nabla_{\bm{X}} \bm{Y} , \bm{Z} ) = 0$ for all $\bm{X}, \bm{Y} \in \Gamma ( \mcN^\perp )$, $\bm{Z} \in \Gamma(\mcN)$.
\end{itemize}

Reality conditions imposed on the top of these features will also yield various geometric interpretations of a null structure in terms of congruences of \emph{real} curves. With these considerations in mind, we shall prove ultimately the following theorem.
\begin{thm}\label{thm-GS-intro}
Let $(\mcM, \bm{g})$ be a three-dimensional Lorentzian manifold satisfying the topological massive gravity equations. Then the tracefree Ricci tensor  is algebraically special if and only if $(\mcM,\bm{g})$ admits a divergence-free congruence of null geodesics (i.e. it is a Kundt spacetime) or a shear-free congruence of timelike geodesics.
\end{thm}
In fact, it is shown in reference \cite{Chow2010} that a Kundt spacetime that is also a solution of the topological massive gravity equations must be algebraically special, while the converse is left as an open problem. Theorem \ref{thm-GS-intro} thus gives an answer to this question.

The strategy adapted in this paper is to use a Newman-Penrose formalism and the Petrov classification of the tracefree Ricci tensor in dimension three. While these have already been used in \cite{Milson2013}, we develop these tools from scratches, and our conventions will certainly differ. Theorem \ref{thm-GS-intro} will in fact follow from more general theorems that we shall prove in the course of the article. A number of solutions of the topological massive gravity equations have been discovered in recent years, see \cites{Chow2010a,Chow2010} and references therein. In a subsequent paper, we shall give further explicit algebraically special solutions of the topological massive gravity equations.

The structure of the paper is as follows. In section \ref{sec-geom-cons}, we set up the scene with a short introduction of the Newman-Penrose formalism, and we review the background on the general geometric properties of null structures on (pseudo-)Riemannian manifolds in dimension three. Particularly relevant here are the notions of co-integrable and co-geodetic null structures of Definition \ref{defn-geom-prop}, the latter being central to the Goldberg-Sachs theorem.

This is taken further in section \ref{sec-real-metrics} where we examine the consequences of the reality conditions on a null structure, which may be understood as a congruence of curves that are either null or timelike. Propositions \ref{prop-co2timegeod} and \ref{prop-co2nullgeod} in particular give real interpretations of co-integrable and co-geodetic null structures.

Section \ref{sec-alg-class} focuses on the algebraic classification of the tracefree Ricci tensor. We introduce the definition of algebraically special Ricci tensors in Definition \ref{defn-alg-spec} based on the notion of principal null structure of Definitions \ref{defn-principal} and \ref{defn-multiple}. This leads to definitions of the complex Petrov types in section \ref{sec-Petrov-types-complex}, and their real signature-dependent analogues in sections \ref{sec-Petrov-types-Euc} and \ref{sec-Petrov-types-Lor}.

Curvature conditions for the existence of co-geodetic and parallel null structures are given in Proposition \ref{prop-integrability-cond} of section \ref{sec-curvature-cond}.

The main results of this paper are contained in section \ref{sec-GS}. We initially give general results for a metric of any signature. We first give in Proposition \ref{prop-obstruction-GS} obstructions for a multiple principal null structure to be co-geodetic in terms of the Cotton tensor and the derivatives of the Ricci scalar. We then show in Theorems \ref{thm-GS-II-gen}, \ref{thm-GS-III-gen}, \ref{thm-GS-N-gen} and \ref{thm-GS-D-gen} how the various algebraically special Petrov types guarantee the existence of a co-geodetic null structure. The converse, that a co-geodetic null structure implies algebraic speciality, is given in Theorem \ref{thm-GS-hard}. The application to topological massive gravity in Theorem \ref{thm-GS-TMG} then follows naturally. The section is wrapped up by giving real versions \ref{thm-Lor-GS-TMG} and \ref{thm-Riem-GS-TMG} of the Goldberg-Sachs theorem.

We end the paper with three appendices. Appendix \ref{app-spinor-calculus} contains a spinor calculus in three dimensions, which we then apply in Appendix \ref{sec-NP} to derive a Newman-Penrose formalism adapted to a null structure. Finally, in Appendix \ref{sec-GS-spinor} we have given alternative, manifestly invariant, proofs of the main theorems of section \ref{sec-GS} in the language of spinors.

\paragraph{Acknowledgments} This work was supported by the Polish National Science Center (NCN) via DEC-2013/09/B/ST1/01799. One of the authors (AT-C) has benefited from an Eduard \v{C}ech Institute postdoctoral fellowship GPB201/12/G028, and a GA\v{C}R (Czech Science Foundation) postdoctoral grant GP14-27885P. He would also like to thank the Centrum Fizyki Teoretycznej PAN for hospitality and financial support during his stay in Warsaw in the period 17-24 January 2015.

\section{Geometric considerations}\label{sec-geom-cons}
Throughout this section, we consider an oriented three-dimensional (pseudo-)Riemannian smooth manifold $(\mcM,\bm{g})$. We shall make use of the abstract index notation of \cite{Penrose1984}. Upstairs and downstairs lower case Roman indices will refer to vector fields and $1$-forms on $\mcM$ respectively, e.g. $V^a$ and $\alpha_a$, and similarly for more general tensor fields, e.g. $T \ind{_{ab}^c_d}$. Symmetrisation will be denoted by round brackets around a set of indices, and skew-symmetrisation by squared brackets, e.g. $A_{(ab)} = \frac{1}{2} \left( A_{ab} + A_{ba} \right)$ and $B_{[ab]} = \frac{1}{2} \left( B_{ab} - B_{ba} \right)$. Indices will be lowered and raised by the metric $g_{ab}$ and its inverse $g^{ab}$ whenever needs arise, e.g. $V_a = V^b g_{ba}$ and $\alpha^a = g^{ab} \alpha_b$, etc... Bold font will be used for vectors, forms and tensors whenever the index notation is suspended.

The space of sections of a given vector bundle $E$, say, over $\mcM$, will be denoted $\Gamma(E)$. The Lie bracket of two vector fields $\bm{V}$ and $\bm{W}$ will be denoted by $[ \bm{V} ,\bm{W}]$.

The orientation on $\mcM$ will be given by a volume form $e_{abc}$ on $\mcM$ satisfying the normalisation conditions
\begin{align}\label{eq-vol-form^2}
e_{abc} e^{def} & = (-1)^q 6 \, g \ind*{_{[a}^d} g \ind*{_b^e} g \ind*{_{c]}^f} \, , &
e_{abe} e^{cde} & = (-1)^q 2 \, g \ind*{_{[a}^c} g \ind*{_{b]}^d} \, , & e_{acd} e^{bcd} & = (-1)^q 2 \, g \ind*{_a^b} \, , &
e_{abc} e^{abc} & = (-1)^q 6 \, ,
\end{align}
where $q$ is the number of negative eigenvalues of the metric $g_{ab}$. We can then eliminate $2$-forms in favour of $1$-forms by means of the Hodge duality operation, i.e.
\begin{align*}
(*\alpha)_a & := \frac{1}{2} e \ind{_a^{bc}} \alpha_{bc} \, ,
\end{align*}
for any $2$-form $\alpha_{ab}$.

The Levi-Civita connection of $g_{ab}$, i.e. the unique torsion-free connection preserving $g_{ab}$, will be denoted $\nabla_a$. The Riemann curvature tensor associated with $\nabla$ is defined by
\begin{align*}
 R \ind{_{a b d} ^c} V^d := 2 \, \nabla_{\lb{a}} \nabla_{\rb{b}} V^c \, .
\end{align*}
In three dimensions, the Riemann tensor decomposes as
\begin{align}\label{eq-Riemann}
 R _{ a b c d } & = 4 \, g _{\lb{a} | \lb{c}} \Phi _{\rb{d}|\rb{b}} + \frac{1}{3} \, R \, g _{\lb{a} | \lb{c}} g _{\rb{d}|\rb{b}} \, ,
\end{align}
where $\Phi_{a b} := R_{ab} - \frac{1}{3} R g_{ab}$ and $R$ are tracefree part of the Ricci tensor $R \ind{_{ab}} := R \ind{_{acb}^c}$ and the Ricci scalar respectively. In three dimensions, the Bianchi identity $\nabla_{[a} R _{bc]de} = 0$ is equivalent to the contracted Bianchi identity
\begin{align}\label{eq-Bianchi-contracted}
\nabla^b \Phi \ind{_{ba}} - \frac{1}{6} \nabla _a R & = 0 \, .
\end{align}
For future use, we define the \emph{Schouten} or \emph{Rho} tensor
\begin{align*}
\Rho_{ab} & := - \Phi_{ab} - \frac{1}{12} R g_{ab} \, .
\end{align*}
To eliminate the use of fractions, we also set
\begin{align*}
S & := \frac{1}{12} R \, .
\end{align*}
With this convention, the \emph{Cotton tensor} takes the form
\begin{align}\label{eq-Cotton}
 A_{a b c} & := 2 \, \nabla_{[b} \Rho _{c] a} = - 2 \, \nabla_{[b} \Phi_{c]a} + 2 \, g_{a[b} \, \nabla_{c]} S \, .
\end{align}
By construction, the Cotton tensor $A_{abc}$ satisfies the symmetry $A_{[a b c]} = 0$ and $A \ind{^a _{a b}} = 0$. It also satisfies the condition
\begin{align}\label{eq-divCotton}
\nabla^a A_{abc} & = 0 \, ,
\end{align}
since commuting the covariant derivatives and using \eqref{eq-Bianchi-contracted} and \eqref{eq-Riemann} give
\begin{align*}
\nabla^a A_{abc} & = - 2 \, \nabla^a \nabla_{[b} \Phi_{c]a} + \cancel{ 2 \, \nabla_{[b} \, \nabla_{c]} S} = \cancel{- 4 \, \nabla_{[b} \nabla_{c]} S } + \cancel{2 \, R \ind{^a_{[bc]}^d} \Phi_{da}} + \cancel{2 \, \Phi \ind{_{[b}^d} \Phi_{|c]d} } = 0 \, .
\end{align*}
It is more convenient to Hodge-dualise $A_{abc}$ to obtain the tensor of valence $2$
\begin{align}\label{eq-Cotton-dual}
(*A)_{ab} & := \frac{1}{2} e \ind{_b^{cd}} A \ind{_{acd}} \, .
\end{align}
Dualising a second time over $ab$ yields $e \ind{_a^{bc}} (*A)_{bc} = A \ind{^b_{ba}} = 0$, which means that $(*A)_{ab}=(*A)_{(ab)}$. The Cotton tensor can then be expressed by
\begin{align*}
(*A)_{ab} & = - e \ind{_{(a|}^{cd}} \nabla_{c} \Phi _{d |b)} \, .
\end{align*}

\subsection{Null structures in dimension three}
For the time being, we shall keep the discussion general, by considering an oriented three-dimensional (pseudo-)Riemannian manifold $(\mcM,\bm{g})$ of any signature. We shall denote the complexification of the tangent bundle by $\Tgt^\C \mcM$. It will be often convenient to consider a complex-valued metric $\bm{g}^\C$ by extending $\bm{g}$ via
\begin{align*}
\bm{g}^\C ( \bm{X}+ \ii \bm{Y} , \bm{Z}+ \ii \bm{W} ) & = \bm{g} ( \bm{X} , \bm{Z} ) - \bm{g} ( \bm{Y} , \bm{W} ) + \ii \left( \bm{g} ( \bm{X} , \bm{W} ) + \bm{g} ( \bm{Y} , \bm{Z} ) \right) \, ,
\end{align*}
for all $\bm{X}, \bm{Y}, \bm{Z}, \bm{W} \in \Gamma( \Tgt \mcM)$. For the remaining part of the paper, we shall omit the ${}^\C$ on $\bm{g}^\C$. It should be clear from the context whether we are using $\bm{g}$ or $\bm{g}^\C$.

Let $\mcN$ be a line subbundle of the complexified tangent bundle $\Tgt^\C \mcM := \C \otimes \Tgt \mcM$, null with respect to the complexified metric, i.e. for any section $k^a$ of $\mcN$, $g_{ab} k^a k^b = 0$, and  $\mcN^\perp$ the orthogonal complement of $\mcN$ with respect to $g_{ab}$, i.e. at any point $p$ in $\mcM$,
\begin{align*}
\mcN^\perp_p & := \left\{ V^a \in \Tgt^\C_p \mcM : V^a k^b g_{ab} = 0 \, , \mbox{for all $k^a \in \mcN_p$} \right\} \, .
\end{align*}
We thus have a filtration $\mcN \subset \mcN^\perp$ on $\Tgt^\C \mcM$.

\begin{defn}\label{defn-almost-null}
We shall refer to a real or complex null line distribution on $(\mcM,\bm{g})$ as a \emph{null structure}.
\end{defn}

This is a three-dimensional specialisation of the notion of \emph{almost null structures} presented in \cites{Taghavi-Chabert2012} (also referred to as \emph{$\gamma$-plane distributions} in \cite{Taghavi-Chabert2013}). Since $\mcN$ is one-dimensional, $\mcN$ is clearly automatically formally integrable, i.e. $[ \Gamma ( \mcN ) , \Gamma ( \mcN ) ] \subset \Gamma ( \mcN )$, and we can unambiguously dispense with the word `almost' in Definition \ref{defn-almost-null}.

\subsection{Rudiments of Newman-Penrose formalism}
For most of the paper, it will be convenient to use the Newman-Penrose formalism as described in details in Appendix \ref{sec-NP}. To do this we introduce a frame of $\Tgt^\C \mcM$ as follows. We fix a section $k^a$ of a given null structure $\mcN$, and choose a vector field $\ell^a$ such that $k^a \ell^b g_{ab} = 1$. Clearly, $\ell^a$ is transversal to $\mcN^\perp$. We also choose a section $n^a$ of $\mcN^\perp$, not in $\mcN$, which we may normalise as $n^a n_a = -\frac{1}{2}$.

\begin{defn}
We say that a frame $( k^a , \ell^a , n^a )$ is \emph{adapted to a null structure $\mcN$ on $(\mcM,\bm{g})$} if and only if $k^a$ generates $\mcN$ and the metric can be expressed as
\begin{align}\label{eq-complex-metric}
g_{ab} & = 2 \, k_{(a} \ell_{b)} - 2 \, n_a n_b \, .
\end{align}
\end{defn}

In fact, we have a class of frames adapted to the null structure, and any two frames in that class are related via the transformation
\begin{align}\label{eq-adapted-modulo}
k^a & \mapsto a k^a \, , &
n^a & \mapsto b \left( n^a + z k^a \right) \, , &
\ell^a & = a^{-1} \left( \ell^a + 2 z n^a + z^2 k^a \right) \, .
\end{align}
for some functions $a$, $z$ and $b$ where $a$ is non-vanishing and $b^2=1$.

We can expand the covariant derivatives of the adapted frame vectors as follows
\begin{align}
\nabla_a k^b & = 2 \, \gamma \, k_a k^b + 2 \, \epsilon \, \ell_a k^b - 4 \, \alpha \, n_a k^b - 2 \, \tau \, k_a n^b - 2 \, \kappa \, \ell_a n^b + 4 \, \rho \, n_a n^b \, , \label{eq-nablak} \\
\nabla_a \ell^b & = - 2 \, \epsilon \, \ell_a \ell^b - 2 \, \gamma \, k_a \ell^b + 4 \, \alpha \, n_a \ell^b + 2 \, \pi \, \ell_a n^b + 2 \, \nu \, k_a n^b - 4 \, \mu \, n_a n^b \, , \label{eq-nablal}\\
\nabla_a n^b & = - \kappa \, \ell_a \ell^b + \nu k_a k^b + \pi \, \ell_a k^b - \tau \, k_a \ell^b + 2 \, \rho \, n_a \ell^b - 2 \, \mu \, n_a k^b \, , \label{eq-nablau}
\end{align}
where $\alpha$, $\gamma$, $\epsilon$, $\kappa$, $\mu$, $\nu$, $\pi$, $\rho$, and $\tau$ are the connection coefficients, also known as \emph{Newman-Penrose coefficients}. We shall also introduce the following notation for the frame derivatives:
\begin{align*}
D & := k^a \nabla_a \, , & \Delta & := \ell^a \nabla_a \, , & \delta & := n^a \nabla_a \, .
\end{align*}
The curvature components of $\nabla_a$ and the Bianchi identies can then be expressed in terms of these coefficients and derivatives thereof, and their full description, also given in Appendix \ref{sec-NP}, is known as the \emph{Newman-Penrose equations}.

\subsection{Geometric properties}
While a null structure $\mcN$ is always integrable, the following definition gives additional geometric conditions that $\mcN$ can satisfy -- these are related to the notion of \emph{intrinsic torsion} examined in \cite{Taghavi-Chabert2013}.

\begin{defn}\label{defn-geom-prop}
Let $\mcN$ be a null structure on $(\mcM,\bm{g})$. We say that
\begin{itemize}
\item $\mcN$ is \emph{co-integrable} if its orthogonal complement $\mcN^\perp$ is formally integrable, i.e.
\begin{align}\label{integrableN}
[ \Gamma ( \mcN^\perp ) , \Gamma ( \mcN^\perp ) ] & \subset \Gamma ( \mcN^\perp ) \, ;
\end{align}
\item  $\mcN$ is \emph{co-geodetic} if its orthogonal complement $\mcN^\perp$ is formally totally geodetic, i.e.
\begin{align}\label{geodetic-Nperp}
\bm{g} ( \nabla_{\bm{X}} \bm{Y} , \bm{Z} ) & = 0 \, , & \mbox{for all $\bm{X}, \bm{Y} \in \Gamma ( \mcN^\perp )$, $\bm{Z} \in \Gamma(\mcN)$;}
\end{align}
\item  $\mcN$ is \emph{parallel} if $\nabla_{\bm{Y}} \bm{X} \in \Gamma(\mcN)$ for all $\bm{X} \in \Gamma(\mcN)$, $\bm{Y} \in \Gamma (\Tgt \mcM)$.
\end{itemize}
\end{defn}

Using the standard formula $[ \bm{X},\bm{Y}] = \nabla_{\bm{X}} \bm{Y} - \nabla_{\bm{Y}} \bm{X}$ for any vector fields $\bm{X}, \bm{Y} \in \Gamma (\Tgt^\C \mcM)$, one can prove the following lemma \cite{Taghavi-Chabert2013}.
\begin{lem}
Let $\mcN$ be a null structure on $(\mcM,\bm{g})$. Then 
\begin{align*}
\mbox{$\mcN$ is parallel} && \Rightarrow && \mbox{$\mcN$ is co-geodetic} && \Rightarrow && \mbox{ $\mcN$ is co-integrable.}
\end{align*}
\end{lem}

\begin{rem}
Note that in three dimensions, unlike in higher odd dimensions, a null structure $\mcN$ automatically satisfies  $\nabla_{\bm{Y}} \bm{X} \in \Gamma(\mcN^\perp)$ for all $\bm{X} \in \Gamma(\mcN)$, $\bm{Y} \in \Gamma(\Tgt \mcM)$, as can be read off from \eqref{eq-nablak}.
\end{rem}

\begin{rem}
Of the three geometric properties listed in Definition \ref{defn-geom-prop}, only the property that $\mcN$ be co-integrable is conformally invariant since it depends only on the Lie bracket. The remaining properties break conformal invariance -- see \cite{Taghavi-Chabert2013} for details.
\end{rem}

It is convenient to re-express condition \eqref{integrableN} and \eqref{geodetic-Nperp} in terms of the Levi-Civita connection as given in the following proposition.
\begin{prop}\label{prop-intrinsic-torsion}
Let $\mcN$ be a null structure on $(\mcM,\bm{g})$, and let $k^a$ be a generator of $\mcN$. Then
\begin{align}
\mbox{$\mcN$ is co-integrable}
& & \Longleftrightarrow & & k _{[a} \nabla_b k_{c]}  & = 0 
& & \Longleftrightarrow & \left(  k^b \nabla_b k^{[a} \right) k^{b]} & = 0 \, , \label{eq-equiv-integrable} \\
\mbox{$\mcN$ is co-geodetic}
& & \Longleftrightarrow & & 
k^a \nabla_b k^b - k^b \nabla_b k^a & = 0
& & \Longleftrightarrow &
k_{[a} \left( \nabla_{b]} k_{[c} \right) k_{d]} & = 0 \, . \label{eq-equiv-strongly}
\end{align}
\end{prop}

\begin{rem}
 Conditions \eqref{eq-equiv-integrable} tell us that any generator of $\mcN$ is formally geodetic, or equivalently in three dimensions, formally twist-free, or equivalently, formally shear-free, i.e. $\mcL_k g_{ab} \propto g_{ab} \pmod{ k_{(a} \alpha_{b)}}$.
 Conditions \eqref{eq-equiv-strongly} tell us that any generator of $\mcN$ is formally geodetic and divergence-free.
\end{rem}

\begin{proof}
Let $k^a$ be a generator of $\mcN$. In terms of the Newman-Penrose coefficients, we have
\begin{align*}
\mbox{$\mcN$ is co-integrable}
& & \Longleftrightarrow & & \kappa = 0 \, ,
\end{align*}
which follows from
\begin{align*}
[ D , \delta ] & = 2 \, \rho \delta 
 + \left( \pi - 2 \, \alpha \right) D - \kappa \Delta \, . \tag{\ref{commutator0001}}
 \end{align*}
On the other hand, comparison with
\begin{align}
\left( D k^{[a} \right) k^{b]} & = - 4 \, \kappa \, n^{[a} k^{b]} \, , & k _{[a} \nabla_b k_{c]} & = - 4 \,  \kappa \, k_{[a} \ell_b n_{c]} \, , \nonumber
\end{align}
establishes the equivalence \eqref{eq-equiv-integrable}.

Further, since \eqref{geodetic-Nperp} can be expressed as $n^a D k_a  = 2 \, \kappa$ and $n^a \delta k_a = 2 \, \rho$ in our null basis, we have
\begin{align*}
\mbox{$\mcN$ is co-geodetic}
& & \Longleftrightarrow & & \kappa = \rho = 0 \, .
\end{align*}
Comparison with 
\begin{align*}
k^a \nabla_b k^b - D k^a & = 4 \left( \kappa \, n^a - \rho \, k^a \right) \, , 
\end{align*}
establishes the equivalence \eqref{eq-equiv-strongly}.
\end{proof}

Another way to express the condition for a null structure to be co-geodetic is given by the following proposition.
\begin{prop}\label{prop-closed-coclosed}
Locally, there is a one-to-one correspondence between closed and co-closed complex-valued $1$-forms and co-geodetic null structures.
\end{prop}

\begin{proof}
We note that a co-integrable null structure $\mcN$ is equivalent to its generator $k^a$, say, satisfying $k_{[a} \nabla_b k_{c]} =0$ since $k_a$ is also the annihilator of $\mcN^\perp$. We can always rescale $k^a$ to make it closed, i.e. $\nabla_{[b} k_{c]} =0$ -- for details, see Lemma 5.1 in \cite{Hill2009}. In this case, $k^a$ also satisfies $k^a \nabla_a k^b =0$. If in addition $\mcN$ is co-geodetic, then using the equivalence \eqref{eq-equiv-strongly}, we have $\nabla^a k_a =0$, i.e. $k^a$ is coclosed. The converse is also true.
\end{proof}

A generalisation to higher dimensions applicable to higher valence spinor fields in the analytic case is given in \cite{Taghavi-Chabert2013}.

As we shall see in the next section, the differential conditions on a null structure $\mcN$ will yield quite different geometric interpretations depending on the signature of $\bm{g}$.

\subsection{Relation to harmonic morphisms}
\begin{defn}[\cites{Baird1988,Baird1995}]
Let $\varphi : \mcM \rightarrow \C$ be a complex-valued smooth map on $\mcM$. We say that $\varphi$ is \emph{horizontal conformal} if it satisfies $(\nabla^a \varphi ) ( \nabla_a \varphi  ) = 0$, and a \emph{harmonic morphism} if it satisfies $( \nabla^a \varphi ) ( \nabla_a \varphi ) = \nabla^a \nabla_a \varphi = 0$.
\end{defn}
With reference to the proof of Proposition \ref{prop-closed-coclosed}, we can rescale a generator $k^a$ of a co-integrable null structure $\mcN$ such that $k_a = \nabla_a \varphi$. Since such a $k^a$ is null, $\varphi$ is a horizontal conformal map. If $\mcN$ is also co-geodetic, $k_a$ is also co-closed and $\varphi$ must be a harmonic morphism. Summarising,
\begin{cor}
On a three-dimensional (pseudo-)Riemannian manifold, locally, there is a one-to-one correspondence between
\begin{itemize}
\item horizontal conformal maps and co-integrable null structures;
\item harmonic morphisms and co-geodetic null structures.
\end{itemize}
\end{cor}

\section{Real metrics}\label{sec-real-metrics}
As before, $(\mcM,\bm{g})$ will denote an oriented three-dimensional (pseudo-)Riemannian manifold. So far the discussion has been independent of the signature of the metric $\bm{g}$. Different metric signatures will induce different reality conditions on $\Tgt^\C \mcM$, and consequently, different geometric interpretations of a null stucture $\mcN$. As is standard, the complex conjugate of $\mcN$ will be denoted $\overline{\mcN}$. We start with the following definition.

\begin{defn}[\cites{Kopczy'nski1992,Kopczy'nski1997}]
The \emph{real index} of a null structure $\mcN$ on $(\mcM,\bm{g})$ at a point $p$ is the dimension of the intersection $\mcN_p \cap \overline{\mcN}_p$.
\end{defn}

\begin{lem}[\cites{Kopczy'nski1992,Kopczy'nski1997}]
At any point, the real index of a null structure $\mcN$ on $(\mcM,\bm{g})$ must be
\begin{itemize}
\item $0$ when $\bm{g}$ has signature $(3,0)$;
\item either $0$ or $1$ when $\bm{g}$ has signature $(2,1)$.
\end{itemize}
\end{lem}

\subsection{Lorentzian signature}
Assume that $(\mcM , \bm{g})$ is an oriented Lorentzian manifold with signature $(2,1)$. The geometrical features of a null structure of real index $0$ or real index $1$ are discussed separately.

\subsubsection{Real index $0$}\label{sec-Lor0}
We shall exhibit the relation between null structures of real index $0$ and congruences of timelike curves. In particular, we shall also show that a \emph{co-geodetic} null structure is equivalent to the existence of a \emph{shear-free} congruence of timelike \emph{geodesics}.

Assume that $\mcN$ is of real index $0$. Then $\mcN$ is a complex distribution, whose real and imaginary parts span a spacelike distribution $\Re( \mcN \oplus \overline{\mcN})$. The orthogonal complement of $\Re( \mcN \oplus \overline{\mcN})$ is necessarily timelike and orthogonal to both $\mcN$ and $\overline{\mcN}$. Let $u^a$ be a unit timelike vector field generating this timelike distribution, i.e. $u^a$ satisfies $u^a u_a = -1$. There is a sign ambiguity in the definition of $u^a$, which can be fixed by a choice of time-orientation. In this case, the metric takes the form
\begin{align}\label{eq-al-contact-Lorentz-metric}
g_{ab} & = h_{ab} - u_a u_b \, .
\end{align}
where $h_{ab}$ is annhilated by $u^a$, i.e. $h_{ab} u^a =0$. The orientation of $(\mcM,\bm{g})$ given by the volume form $e_{abc}$ normalised as in \eqref{eq-vol-form^2} with $q=1$ allows us to Hodge-dualise
the timelike vector field $u^a$ to produce a $2$-form $\omega_{ab}$, i.e.
\begin{align*}
\omega_{ab} := - e _{abc} u^c \, .
 \end{align*}
This in turns yields an endomorphism
\begin{align*}
 J \ind{_a^b} & = \omega_{ac} g^{cb} \, ,
\end{align*}
on $\Tgt \mcM$, which can be seen to satisfy
\begin{align}\label{eq-contact-Lor-properties}
J \ind{_a^c} J \ind{_c^b} & = - h \ind{_a^b} \, , & u^b J \ind{_b^a} & = 0 \, .
\end{align}
This yields a splitting of the complexified tangent bundle
\begin{align*}
\Tgt^\C  \mcM & = \Tgt^{(1,0)} \oplus \Tgt^{(0,1)} \oplus \Tgt^{(0,0)} \, ,
\end{align*}
where $\Tgt^{(1,0)}$, $\Tgt^{(0,1)}$ and $\Tgt^{(0,0)}$ are the $-\ii$-, $+\ii$- and $0$-eigenbundles of $J \ind{_a^b}$ respectively. In particular, a null structure of index $0$ yields a CR structure with a preferred splitting as described in \cite{Hill2009}.

In an adapted frame $( k^a , \ell^a , n^a )$, we have $u^a = \sqrt{2} n^a$ with the following reality conditions
 \begin{align*}
 ( k^a , \ell^a , n^a ) \mapsto  ( \overline{k^a} , \overline{\ell^a} , \overline{n^a} ) = ( \ell^a , k^a , n^a  ) \, .
 \end{align*}
Thus $\ell^a$ is the complex conjugate of $k^a$, and $n^a$ is real. If our orientation is chosen such that
\begin{align*}
e_{abc} & = 6 \ii k_{[a} \ell_b u_{c]} \, ,
\end{align*}
then $k^a$ and $\bar{k}^a := \ell^a$ satisfy
\begin{align*}
k^b J \ind{_b^a} & = \ii k^a \, , & \bar{k}^b J \ind{_b^a} & = - \ii \bar{k}^a \, .
\end{align*}

At this stage, we remark that a unit timelike vector field determines a congruence of oriented timelike geodesics, and conversely, given such a congruence, we can \emph{always} find a unit timelike vector field $u^a$ tangent to the curves of the congruence. Further, the effect of changing the orientation of $u^a$ will have the effect of interchanging the null structure $\mcN$ and its complex conjugate $\overline{\mcN}$. We can therefore summarise our results in the following way.
\begin{prop}
On an oriented and time-oriented three-dimensional Lorentzian manifold, there is a one-to-one correspondence between  null structures and congruences of oriented timelike curves.
\end{prop}

It is convenient to encode the geometric properties of $\mcN$ and $\mcN^\perp$ in the covariant derivative of $u^a$. To characterise these properties, we decompose the covariant derivative of $u^a$ into irreducibles under the stabiliser of $u^a$ in $\SO(2,1)$ as recorded in the following definition.
\begin{defn} Let $u^a$ be a unit timelike vector field on $(\mcM,\bm{g})$ and $\mcU$ its associated congruence of oriented (timelike) curves. Then $\mcU$ is 
\begin{itemize}
\item \emph{geodetic} if and only if
\begin{align}\label{eq-geodesy-sp}
u^b \nabla_b u^a & = 0 \, ;
\end{align}
\item \emph{twist-free} if and only if
\begin{align}\label{eq-twist-sp}
 \left( \nabla_{[a} u_b \right) u_{c]} & = 0 \, ;
\end{align}
\item \emph{divergence-free} if and only if
\begin{align}\label{eq-dilation-sp}
\nabla^c u_c & = 0 \, ;
\end{align}
\item \emph{shear-free} if and only if
\begin{align}\label{eq-shear-sp}
\nabla_{(a} u_{b)} - \frac{1}{2} h_{ab} \nabla^c u_c + \left( u^c \nabla_c u_{(a} \right) u_{b)} & = 0 \, .
\end{align}
\end{itemize}
\end{defn}

Using the standard formula for the Lie derivative
\begin{align}\label{eq-shear-free-Lie2}
\mcL_u h_{ab} & = 2 \, \nabla_{(a} u_{b)} + 2 \left( u^c \nabla_c u_{(a} \right) u_{b)} \, ,
\end{align}
we can re-express the shear-free condition \eqref{eq-shear-sp} as
\begin{align}\label{eq-shear-free-Lie}
\mcL_u h_{ab} & = \Omega^2 \, h_{ab} \, ,
\end{align}
for some function $\Omega$.

These properties of the congruence $\mcU$ do not depend on the orientation of $u^a$, and thus also apply to congruences of unoriented curves.

\begin{prop}\label{prop-co2timegeod}
Let $\mcN$ be a null structure of real index $0$ on $(\mcM,\bm{g})$ equipped with a time-orientation, and let $\mcU$ be its associated congruence of oriented timelike curves. Then
\begin{itemize}
\item $\mcN$  is co-integrable if and only if $\mcU$ is shear-free;
\item $\mcN$  is co-geodetic if and only if $\mcU$ is shear-free and geodetic.
\end{itemize}
\end{prop}

\begin{proof}
In terms of the Newman-Penrose formalism, and with suitable reality conditions, we have
\begin{align}\label{eq-covderv-u}
\nabla_a u_b & = - \sqrt{2} \left( \bar{\kappa} k_a k_b + \kappa \bar{k}_a \bar{k}_b \right) + 2 \left( \bar{\rho} u_a k_b + \rho u_a \bar{k}_b \right)  + \frac{\ii}{\sqrt{2}} ( \tau - \bar{\tau} ) \omega_{ab} - \frac{1}{\sqrt{2}} (\tau + \bar{\tau}) h_{ab} \, ,
\end{align}
so that taking the irreducible components, we obtained
\begin{align*}
u^b \nabla_b u_a & = - 2 \left( \rho \, \bar{k}_a + \bar{\rho} \, k_a \right) \, , \\
\left( \nabla_{[a} u_b \right) u_{c]} & = \frac{\ii}{\sqrt{2}} \left(  \tau - \bar{\tau} \, \right) \omega_{[ab} u_{c]} \, , \\
\nabla^a u_a & = - \sqrt{2} \left( \tau -\bar{\tau} \right) \, , \\
2 \nabla_{(a} u_{b)} - h_{ab} \nabla^c u_c + 2 \left( u^c \nabla_c u_{(a} \right) u_{b)} & = - 2 \sqrt{2} \left( \kappa \, \bar{k}_a \bar{k}_b + \bar{\kappa} \, k_a k_b \right) \, .
\end{align*}
The conclusion of the proof now follows from (the proof of) Proposition \ref{prop-intrinsic-torsion}.
\end{proof}

\paragraph{Local forms of metrics}
We now give the normal form of a metric admitting a shear-free congruence of timelike geodesics.

\begin{prop}\label{prop-normal-metric}
Let $(\mcM,\bm{g})$ be an oriented and time-oriented three-dimensional Lorentzian manifold admitting a co-integrable null structure $\mcN$, $k^a$ a section of $\mcN$ and $u^a$ the associated unit timelike vector field. Then, around each point, there exist coordinates $(t,z,\bar{z})$ such that the metric takes the form
\begin{align}
\bm{g} & = - \left( \dd t -  \bar{p}  \, \dd z - p \, \dd \bar{z}  \right)^2 + 2 h^2 \dd z \, \dd \bar{z} \, , \label{eq-canon-met-SF} \\
\bm{k} & = h^{-1} \left( \partial_{\bar{\zeta}} + p \partial_t \right) \, , & \bm{u} & = \partial_t  \, , \nonumber
\end{align}
where $h=h(z,\bar{z},t)$ and $p=p(z,\bar{z},t)$.

If $\mcN$ is co-geodetic, we have in addition
\begin{align}\label{eq-canon-met-Geod}
\partial_t p & = 0 \, .
\end{align}
\end{prop}

\begin{proof}
We first note that the complex-valued $1$-form $k_a$ satisfies $k_{[a} \nabla_b k_{c]} =0$. By a lemma of reference \cite{Hill2008}, $k_a$ can be put in the form $k_a = h \nabla_a \zeta$ for some real function $h$ and complex function $\zeta$ such that $\dd \zeta \wedge \dd \bar{\zeta} \neq 0$. We can therefore use $\zeta$ and its complex conjugate $\bar{\zeta}$ as complex coordinates on $(\mcM,\bm{g})$. We can also choose a real coordinate $t$ such that $u^a\nabla_a = \parderv{}{t}$. The metric must therefore take the form \eqref{eq-canon-met-SF}.

Further, the property that $u^a$ be geodetic can be expressed as
\begin{align*}
0 & = - u^b \nabla_b u_a = u^b \left(\nabla_a u_b - \nabla_b u_a \right) \, ,
\end{align*}
which leads to \eqref{eq-canon-met-Geod}.
\end{proof}

\begin{rem}\label{rem-contact-Lor}
One can check that a null structure $\mcN$ of real index $0$ on a three-dimensional Lorentzian manifold $(\mcM,\bm{g})$ is also known as an \emph{almost contact Lorentzian structure} -- see \cite{Calvaruso2011} for details and generalisation to higher odd dimensions. The contact distribution is precisely annihilated by the timelike vector field $u^a$. When the null structure is co-geodetic, the almost contact Lorentzian structure is said to be \emph{normal}.
\end{rem}

\subsubsection{Real index $1$}\label{sec-Lor1}
A null structure $\mcN$ of real index $1$ satisfies $\dim(\mcN_p \cap \overline{\mcN}_p)=1$ at every point $p \in \mcM$. In particular, since $\mcN$ is one-dimensional, it must be totally real.

We can therefore introduce a totally real basis $( k^a , \ell^a , u^a)$ of $\Tgt \mcM$, where $k^a$ is a generator of $\mcN$, $\ell^a$ a null vector field transversal to $\mcN^\perp$ such that $k^a \ell^b g_{ab} = 1$, and $u^a$ a unit spacelike vector field in $\mcN^\perp$, i.e. $u^a u_a =1$,  and thus complementary to $k^a$ and $\ell^a$. The Lorentzian metric then takes the form
\begin{align}\label{eq-real-metric}
g_{ab} & = 2 \, k_{(a} \ell_{b)} + u_a u_b \, .
\end{align}
Setting $n^a = \frac{\ii}{\sqrt{2}} u^a$, an adapted frame $( k^a , \ell^a , n^a )$ can be recovered from $( k^a , \ell^a , u^a)$, in which case it satisfies the reality condition
 \begin{align*}
( k^a , \ell^a , n^a ) \mapsto  ( \overline{k^a} , \overline{\ell^a} , \overline{n^a} ) = ( k^a , \ell^a , -n^a  ) \, .
 \end{align*}

In what follows, we shall not be concerned with the orientation of $k^a$, i.e. whether it is past-pointing or future-pointing.
\begin{defn}
 Let $\mcN$ be a null structure of real index $1$ on $(\mcM,\bm{g})$, and $k^a$ a section of $\mcN$. Let $\mcK$ be the congruence of null curves generated by $k^a$. Then $\mcK$ is
 \begin{itemize}
 \item \emph{geodetic} if and only if
\begin{align}\label{eq-null-geodetic}
\left( k^b \nabla_b k^{[a} \right) k^{b]} & = 0 \, ;
\end{align} 
 \item \emph{divergence-free geodetic} if and only if
 \begin{align}\label{eq-null-non-expanding}
k^a \nabla_b k^b - k^b \nabla_b k^a & = 0 \, .
 \end{align}
 \end{itemize}
\end{defn}

\begin{defn}
A three-dimensional Lorentzian manifold equipped with a divergence-free congruence of null geodesics is called a \emph{Kundt spacetime}.
\end{defn}

\begin{rem}
In dimensions greater than three, congruences of null geodesics are also characterised by their \emph{shear} and \emph{twist}, and a Kundt spacetime is usually defined as a Lorentzian manifold equipped with a shear-free, twist-free and divergence-free congruence of null geodesics. However, in three dimensions, any congruence of null geodesics has vanishing shear and twist.
\end{rem}

From Proposition \ref{prop-intrinsic-torsion}, we now obtain the geometric interpretation of a null structure of real index $1$.
\begin{prop}\label{prop-co2nullgeod}
 Let $\mcN$ be a null structure $\mcN$ of real index $1$ on $(\mcM,\bm{g})$. Then
 \begin{itemize}
 \item $\mcN$ is co-integrable if and only if it generates a congruence of null geodesics.
 \item $\mcN$ is co-geodetic if and only if it generates a divergence-free congruence of null geodesics.
 \end{itemize}
\end{prop}

\subsection{Euclidean signature}\label{sec-Euc}
Assume now $(\mcM , \bm{g})$ has Euclidean signature. Then a null structure is necessarily of real index $0$, and an adapted frame $\{ k^a , \ell^a , n^a \}$ will satisfy the reality conditions
 \begin{align*}
 \{ k^a , \ell^a , n^a \} \mapsto  \{ \overline{k^a} , \overline{\ell^a} , \overline{n^a} \} = \{ \ell^a , k^a , -n^a  \} \, .
 \end{align*}
 Set $\bar{k}^a := \ell^a$ and $u^a := \sqrt{2} \ii \, n^a$, so that $u^a$ is a real spacelike vector of unit norm, i.e. $u^a u_a = 1$. We can then write the metric \eqref{eq-complex-metric} as
\begin{align}\label{eq-al-contact-metric}
g_{ab} & = 2 \, k_{(a} \bar{k}_{b)} + u_a u_b \, .
\end{align}
Clearly, this setting is almost identical to the case where $(\mcM,\bm{g})$ has Lorentzian signature and is equipped with a null structure of real index $0$. The only difference is that now $u^a$ is spacelike, rather than timelike.
Real tensors $\omega_{ab}$ and $J \ind{_a^b}$ are defined in exactly the same manner as in the Lorentzian case, and we now have $h_{ab}  = 2 \, k_{(a} \bar{k}_{b)} = g_{ab} - u_a u_b$.

The reader is invited to go through section \ref{sec-Lor0} with $u^a$ now spacelike.

\begin{rem}\label{rem-contact-Riem}
Following on from Remark \ref{rem-contact-Lor}, a null structure $\mcN$ on a three-dimensional Riemannian manifold $(\mcM,\bm{g})$ can be shown to be equivalent to an \emph{almost contact Riemannian (or metric) structure} -- see \cite{Chinea1990} and references therein for details.
\end{rem}

\section{Algebraic classification of the tracefree Ricci tensor}\label{sec-alg-class}
A special feature of the Riemann curvature of the Levi-Civita connection on a three-dimensional (pseudo-)Riemannian manifold $(\mcM,\bm{g})$ is that it is entirely determined by the Ricci tensor. Its tracefree part $\Phi _{ab}$ belongs to a five-dimensional irreducible representation of $\SO(3,\C)$, and as for the Weyl tensor in four dimensions, we can introduce the notion of principal null structure to classify $\Phi _{ab}$.

\begin{defn}\label{defn-principal}
Let $\mcN$ be a null structure. We say that it is \emph{principal} at a point $p$ of $\mcM$ if a null vector $\xi^a$ generating $\mcN_p$ satisfies
\begin{align}\label{eq-PND-eq}
\Phi_{ab} \xi^a \xi^b & = 0 \, ,
\end{align}
at $p$.
\end{defn}

Now, the space of all complex null vectors at a point is parametrised by points of the Riemann sphere $S^2 \cong \CP^1$. To be precise, using a standard chart on $\CP^1$, an arbitrary null vector, written in an adapted frame $(k^a,\ell^a,n^a)$, is of the form
\begin{align}\label{eq-arb-null-vec}
\xi^a (z) & = k^a + 2z \, n^a + z^2 \ell^a \, , 
\end{align}
for some $z \in \C$. The other standard chart on $\CP^1$ is simply obtained by sending $z$ to $z^{-1}$ in \eqref{eq-arb-null-vec}. Thus, to determine all the principal null structures at a point, it suffices to plug \eqref{eq-arb-null-vec} into 
\eqref{eq-PND-eq} to form the quartic polynomial
\begin{align}\label{eq-Ricci-polynomial}
0 & = \Phi (z) := \frac{1}{2} \Phi_{ab} \xi^a (z) \xi^b (z) = \Phi _0 + 4 \, \Phi _1 z + 6 \, \Phi _2 z^2 + 4 \, \Phi _3 z^3 + \Phi _4 z^4 \, ,
\end{align}
in $\C$, where
\begin{align*}
 \Phi_0 & := \frac{1}{2} \Phi _{ab} k^a k^b \, , & 
 \Phi_1 & := \frac{1}{2} \Phi _{ab} k^a n^b \, , &
 \Phi_2 & := \frac{1}{2} \Phi _{ab} k^a \ell^a = \frac{1}{2} \Phi _{ab} n^a n^b \, , \\
 \Phi_3 & := \frac{1}{2} \Phi _{ab} \ell^a n^b \, , &
 \Phi_4 & := \frac{1}{2} \Phi _{ab} \ell^a \ell^b \, .
\end{align*}

Thus, a root $z$ of \eqref{eq-Ricci-polynomial} determines a principal null structure $\mcN$ where $\xi^a(z)$ generates $\mcN$ at a point. Conversely, any
principal null structure determines a unique root (up to multiplicity) of \eqref{eq-Ricci-polynomial} at a point. In particular, the algebraic classification of the tracefree Ricci tensor boils down to the classification of the roots of \eqref{eq-Ricci-polynomial} and their multiplicities.  A full review of the classification of the Weyl tensor in four dimensions is given in \cite{Gover2011}, and we shall only recall their results closely following their terminology.

\begin{defn}\label{defn-multiple}
Let $\mcN$ be a principal null structure determined by a root $z$ of the associated polynomial \eqref{eq-Ricci-polynomial} at a point. We say that $\mcN$ is \emph{multiple} at that point if $z$ is multiple.
\end{defn}

\begin{defn}\label{defn-alg-spec}
We say that $\Phi_{ab}$ is \emph{algebraically special} at a point if it admits a multiple principal null structure at that point.
\end{defn}

Rather than considering the quartic polynomial \eqref{eq-Ricci-polynomial}, it is convenient to consider the quartic \emph{homogeneous} polynomial
\begin{align}\label{eq-Ricci-polynomial-spinor}
0 & = \Phi (z) = \Phi _{ABCD} \xi^A \xi^B \xi^C \xi^D \, ,
\end{align}
where $\xi^A$ are now complex homogeneous coordinates on $\CP^1$, and $\Phi _{ABCD}$ is an element of $\odot^4 (\C^2)^*$, with the understanding that the upper case Roman indices take the values $0$ and $1$. We can then recover \eqref{eq-Ricci-polynomial} by setting $\xi^A (z) = o^A + z \, \iota^A$ where $\{o^A, \iota^A\}$ is a basis of $\C^2$. The roots of \eqref{eq-Ricci-polynomial-spinor} then determine a unique factorisation (up to permutation of the factors)
\begin{align*}
0 & = \Phi (z) = (\xi^A \alpha_A) (\xi^B \beta_B) (\xi^C \gamma_C) (\xi^D \delta_D) \, ,
\end{align*}
where $\alpha_A$, $\beta_A$, $\gamma_A$ and $\delta_A$ are elements $(\C^2)^*$ defined up to scale. In this case, we can write
\begin{align*}
\Phi _{ABCD} & \propto \alpha_{(A} \beta_B \gamma_C \delta_{D)} \, .
\end{align*}
Multiplicities of the roots of \eqref{eq-Ricci-polynomial-spinor} will be mirrored by some of the corresponding $\alpha_A$, $\beta_A$, $\gamma_A$ and $\delta_A$ being proportional to each other.

\begin{rem}\label{rem-spinor-Petrov}
The above identification is $\Phi_{ab}$ with $\Phi _{ABCD}$ is a consequence of the local isomorphism of Lie groups $\SO(3,\C) \cong \SL(2,\C)$ as explained in Appendix \ref{app-spinor-calculus}, where $\C^2$ is identified with the spinor representation of $\SO(3,\C)$. This is virtually identical to the treatment of the (anti-)self-dual part of the Weyl tensor in four dimensions \cites{Petrov2000,Witten1959,Penrose1960}.
\end{rem}

\subsection{Complex case}\label{sec-Petrov-types-complex}
If the metric is complex with no preferred reality condition imposed on it, the coefficients $( \Phi _0 , \Phi _1 , \Phi _2 , \Phi _3 , \Phi _4 )$ are generically complex, and the polynomial \eqref{eq-Ricci-polynomial} has generically four distinct roots, and thus four distinct principal null structures. Following the notation of \cite{Penrose1986}, we can encode the multiplicities of the roots of \eqref{eq-Ricci-polynomial} in a partition $\{a_1,a_2,a_3,a_4\}$ of the integer $4$, where $a_1+a_2+a_3+a_4=4$. We shall omit those $a_i$ from the partition whenever $a_i=0$. We thus obtain a Petrov classification of $\Phi_{ab}$:
\begin{displaymath}
\begin{array}{c|c|c}
\text{Petrov type} & \{a_1,a_2,a_3,a_4 \} & \Phi _{ABCD} \\
\hline
\text{I} & \{ 1,1,1,1 \} & \alpha_{(A} \beta_{B} \gamma_{C} \delta_{D)} \\
\text{II} & \{ 2,1,1 \} & \alpha_{(A} \alpha_{B} \beta_{C} \gamma_{D)} \\
\text{III} & \{3,1\} & \alpha_{(A} \alpha_{B} \alpha_{C} \beta_{D)} \\
\text{N} & \{4\} & \alpha_{A} \alpha_{B} \alpha_{C} \alpha_{D} \\
\text{D} & \{2,2\} & \alpha_{(A} \alpha_{B} \beta_{C} \beta_{D)}\\
\text{O} & \{ - \} & 0
\end{array}
\end{displaymath}
Petrov types II, II, III and N single out  a multiple principal null structure, and Petrov type D a pair of distinct multiple principal null structures.

\begin{rem}\label{rem-typeD}
Suppose $\Phi_{ab}$ is of Petrov type D so that the polynomial \eqref{eq-Ricci-polynomial} has two distinct roots  of multiplicity two.  Then we can always arrange that these roots are $0$ and $\infty$ in $\CP^1$, which is done by some suitable change of frame \eqref{eq-adapted-modulo} by assuming, with no loss, that one of these roots is $0$. In this case, we have a distinguished frame $(k^a,\ell^a,n^a)$ adapted to \emph{both} multiple principal null structures, namely $k^a$ and $\ell^a$.
\end{rem}

Let $k^a$ be a generator of a null structure $\mcN$. To verify whether $\mcN$ is a (multiple) principal null structure, it suffices to check whether any of the following algebraic relations holds:
\begin{displaymath}
\begin{array}{ccc}
\text{Petrov type I:} & k^a \Phi_{ab} k^b = 0 \, , & \Phi_0 = 0 \, , \\
\text{Petrov type II:} & k_{[a} \Phi_{b]c} k^c = 0 \, , & \Phi_0 = \Phi_1 = 0 \, , \\
\text{Petrov type D:} & k_{[a} \Phi_{b]c} k^c = 0 \quad \mbox{and} \quad \ell_{[a} \Phi_{b]c} \ell^c = 0 \, , & \Phi_0 = \Phi_1 = 0 = \Phi_3 = \Phi_4 \, , \\
\text{Petrov type III:} & k^a \Phi_{ab} = 0 \quad \mbox{or} \quad k_{[a} \Phi_{b][c} k_{d]} = 0 \, , & \Phi_0 = \Phi_1 = \Phi_2 = 0 \, , \\
\text{Petrov type N:} & k_{[a} \Phi_{b]c} = 0 \, , & \Phi_0 = \Phi_1 = \Phi_2 = \Phi_3 = 0 \, , 
\end{array}
\end{displaymath}
where $\ell^a$, in the case of Petrov type D, determines a multiple principal structure distinct from $\mcN$.

Because of the importance of the Goldberg-Sachs theorem, we highlight the algebraically special condition of the tracefree Ricci tensor by means of the following proposition. In particular, the proofs of the various versions of the Goldberg-Sachs theorem in section \ref{sec-GS} will impinge on it.
\begin{prop}
In an adapted frame, the tracefree Ricci tensor is algebraically special if and only if
\begin{align*}
\Phi_0 = \Phi_1 = 0 \, .
\end{align*}
\end{prop}

\subsection{Real case}
\subsubsection{Euclidean signature}\label{sec-Petrov-types-Euc}
In Euclidean signature, the four roots of the polynomial \eqref{eq-Ricci-polynomial} come in two complex conjugate pairs. Thus, we distinguish only two algebraic types: the generic type G, where the conjugate pairs of complex roots are distinct, and the special type D, where the pairs coincide. Notationally, we shall bracket a conjugate pair of complex roots, i.e. $\{1^\C ,\overline{1^\C} \}$, and where the ${}^\C$ indicates that the root is complex.
\begin{displaymath}
\begin{array}{c|c|c}
\text{Petrov type} & \{a_1,a_2,a_3,a_4 \} & \Phi _{ABCD} \\
\hline
\text{G} & \{ \{1^\C ,\overline{1^\C} \}, \{1^\C ,\overline{1^\C} \} \} & \xi_{(A} \hat{\xi}_{B} \eta_{C} \hat{\eta}_{D)}\\
\text{D} & \{ \{1^\C ,\overline{1^\C} \}^2 \} & \xi_{(A} \hat{\xi}_{B} \xi_{C} \hat{\xi}_{D)}\\
\text{O} & \{ - \} & 0 
\end{array}
\end{displaymath}
Here, a $\hat{}$ denotes a reality  condition on $(\C^2)^*$ defined as follows: if $\xi_A = (\xi_0 , \xi_1)$, then $\hat{\xi}_A = (-\bar{\xi}_1 , \bar{\xi}_0)$.

Since a null structure determines a unit vector $u^a$ (up to sign) and an endomorphism $J \ind{_a^b}$ as described in section \ref{sec-Euc}, we can characterise principal null structures as follows
\begin{displaymath}
\begin{array}{cc}
\text{Petrov type G:} & J \ind{_{(a}^c} \Phi_{b)c} = 0 \quad \mbox{or} \quad 2 u_{[a} \Phi_{b][c} u_{d]} + u_{[a} g_{b][c} u_{d]} u^e u^f \Phi_{ef} = 0 \, , \\
\text{Petrov type D:} & J \ind{_{(a}^c} \Phi_{b)c} = 0 \quad \mbox{and} \quad u_{[a} \Phi_{b]c} u^c = 0 \, .
\end{array}
\end{displaymath}

\begin{rem}
In the context of almost contact metric manifolds (see Remarks \ref{rem-contact-Lor} and \ref{rem-contact-Riem}), the type D condition is equivalent to the manifold being \emph{$\eta$-Einstein}, i.e. the Ricci tensor takes the form $R_{ab} = a \, g_{ab} + b \, u_a u_b$
for some unit vector $u^a$ and functions $a$ and $b$, i.e. $R= 3a -b$ and $\Phi_{ab} = \frac{b}{3} \left( g_{ab} + 3 u_a u_b \right)$.
\end{rem}

\subsubsection{Lorentzian signature}\label{sec-Petrov-types-Lor}
In Lorentzian signature, the situation is a little more complex, and one distinguishes ten Petrov types including type O. Here, a root of \eqref{eq-Ricci-polynomial} can either be real or complex. We distinguish the following cases, excluding type O:
\begin{itemize}
\item If all roots are real, we obtain a totally real analogue of the complex Petrov types with five Petrov types denoted G$_r$, II$_r$, III$_r$, N$_r$ and D$_r$. Petrov types II$_r$, III$_r$ and N$_r$ single out a multiple principal null structure of real index $1$, and Petrov type D$_r$ a pair of distinct multiple principal null structures of real index $1$.
\item If all the roots are complex, they come in conjugate pairs, and we have two Petrov types, G and D as in the Euclidean case. Petrov type D singles out a complex conjugate pair of multiple principal null structures of real index $0$.
\item The remaining types, denoted SG and II, occur when $\Phi(z)$ has two real roots and one conjugate pair of complex roots. Petrov Type II singles out a multiple principal null structure of real index $1$.
\end{itemize}
Using the same notation as above to describe the degeneracy and reality of the roots of \eqref{eq-Ricci-polynomial}, we obtain the following Petrov types of $\Phi_{ab}$:
\begin{displaymath}
\begin{array}{c|c|c}
\text{Petrov type} & \{a_1,a_2,a_3,a_4 \}  & \Phi _{ABCD} \\
\hline
\text{G} & \{ \{ 1^\C ,\overline{1^\C} \}, \{ 1^\C ,\overline{1^\C} \} \} & \xi_{(A} \hat{\xi}_{B} \eta_{C} \hat{\eta}_{D)} \\
\text{SG} & \{ 1,1,\{ 1^\C ,\overline{1^\C} \} \} & \alpha_{(A} \beta_{B} \eta_{C} \hat{\eta}_{D)}\\
\text{II} & \{ 2, \{ 1^\C ,\overline{1^\C} \} \} & \alpha_{(A} \alpha_{B} \eta_{C} \hat{\eta}_{D)}\\
\text{D} & \{ \{ 1^\C ,\overline{1^\C})^2 \} & \xi_{(A} \hat{\xi}_{B} \eta_{C} \hat{\eta}_{D)}\\
\text{G$_r$} & \{ 1,1,1,1 \} & \alpha_{(A} \beta_{B} \gamma_{C} \delta_{D)} \\
\text{II$_r$} & \{ 2,1,1 \} &  \alpha_{(A} \alpha_{B} \beta_{C} \gamma_{D)} \\
\text{III$_r$} & \{ 3,1 \} & \alpha_{(A} \alpha_{B} \alpha_{C} \beta_{D)} \\
\text{N$_r$} & \{ 4 \} & \alpha_{A} \alpha_{B} \alpha_{C} \alpha_{D} \\
\text{D$_r$} & \{ 2,2 \} & \alpha_{(A} \alpha_{B} \beta_{C} \beta_{D)}\\
\text{O} & \{ - \} & 0 
\end{array}
\end{displaymath}
Characterisation of the Petrov types in terms of their principal null structures can be done as in the previous cases in the obvious way.

\section{Curvature conditions}\label{sec-curvature-cond}
Before we move to our main results on the Goldberg-Sachs theorem, we remark that in three dimensions, unlike in higher dimensions \cite{Taghavi-Chabert2013}, the conformally invariant condition
\begin{align*}
[ \Gamma ( \mcN^\perp ) , \Gamma ( \mcN^\perp ) ] \subset \Gamma ( \mcN^\perp ) \, ,
\end{align*}
for a null structure $\mcN$, imposes no constraint on the curvature. Non-trivial constraints on the Ricci curvature are expected to arise from non-conformally invariant conditions on $\mcN$ \cite{Taghavi-Chabert2013}. In particular, we have the following proposition.
\begin{prop}\label{prop-integrability-cond}
Let $\mcN$ be a null structure on an oriented three-dimensional (pseudo-)Riemannian manifold $(\mcM,\bm{g})$. 
\begin{itemize}
\item Suppose that $\mcN$ is co-geodetic. Then $\mcN$ is a principal null structure, i.e.
\begin{align}\label{eq-integrability-cond}
k^a k^b \Phi_{ab} & = 0 \, ,
\end{align}
for any generator $k^a$ of $\mcN$.
\item Suppose $\mcN$ is parallel. Then $\Phi_{ab}$ is algebraically special, i.e.
\begin{align}\label{eq-integrability-cond-parallel}
k^c \Phi_{c[a} k_{b]} & = 0 \, .
\end{align}
for any generator $k^a$ of $\mcN$.
\end{itemize}
\end{prop}

\begin{proof}
We use the Newman-Penrose formalism of appendix \ref{sec-NP} adapted to $\mcN$.
\begin{itemize}
\item By assumption, $\kappa = \rho = 0$. Then, by equation \eqref{Ricci000100}, we have $\Phi _0 = 0$.
\item By assumption, $\kappa = \rho = \tau = 0$. Then, \eqref{Ricci001100}, \eqref{Ricci000100} and \eqref{Ricci110100} give $\Phi _1 = 0$, $\Phi _0 = 0$ and $\Phi _2 = S$ respectively.
\end{itemize}
This completes the proof.
\end{proof}

\begin{rem}
In anticipation of the Goldberg-Sachs theorem, which will be concerned with the relation between algebraically special tracefree Ricci tensors and co-geodetic null structures, the above proposition tells that the existence of a co-geodetic null structure $\mcN$ already imposes algebraic constraints relating the curvature and $\mcN$.
\end{rem}

\section{Main results}\label{sec-GS}
Throughout this section, $(\mcM,\bm{g})$ will denote an oriented three-dimensional (pseudo-)Riemannian manifold. As before, the tracefree Ricci tensor will be denoted by $\Phi_{ab}$ and the Ricci scalar by $R$, and its scalar multiple $S := \frac{1}{12} R$. We also recall the definitions of the Cotton tensor and its Hodge dual:
\begin{align*}
 A_{a b c} & := 2 \, \nabla_{[b} \Rho _{c] a} = - 2 \, \nabla_{[b} \Phi_{c]a} + 2 g_{a[b} \, \nabla_{c]} S \, , & 
 (*A)_{ab} & := \frac{1}{2} e \ind{_b^{cd}} A \ind{_{acd}} = - e \ind{_{(a|}^{cd}} \nabla_{c} \Phi _{d |b)} \, .
\end{align*}
In particular, the components of $A_{abc}$ with respect to the frame $(k^a, \ell^a, n^a)$ will be denoted
\begin{align}
\begin{aligned}\label{eq-A-components}
 A_0 & := 2 \, A _{abc} k^a k^b n^c \, , &
 A_1 & := A _{abc} k^a k^b \ell^c \, , &
 A_2 & := 2 \, A _{abc} k^a n^b \ell^c \, , \\
 A_3 & := A _{abc} \ell^a k^b \ell^c \, , &
 A_4 & := 2 \, A _{abc} \ell^a n^b \ell^c \, .
 \end{aligned}
\end{align}
The results of sections \ref{sec-obst}, \ref{sec-algsp-co-geo} and \ref{sec-co-geo-algsp} will be stated for an arbitrary complex-valued metric with no reality conditions imposed.

\subsection{Obstructions to the existence of multiple co-geodetic null structures}\label{sec-obst}
We first present results concerning curvature obstructions to the existence of multiple co-geodetic null structures.

\begin{prop}\label{prop-obstruction-GS}
Let $\mcN$ be a null structure on $(\mcM,\bm{g})$, and let $k^a$ be any generator of $\mcN$. Suppose that $\mcN$ is co-geodetic and multiple principal.
\begin{itemize}
 \item If $\Phi_{ab}$ is of Petrov type II, then
\begin{align}\label{eq-obstruction-II}
k^a k^b \left( A_{abc} + 3 \, g_{bc} \nabla_a S \right) & = 0 \, .
\end{align}
 \item If $\Phi_{ab}$ is of Petrov type III, then
\begin{align}\label{eq-obstruction-III}
k^a \left( A_{abc} - 2 \, g_{a[b} \nabla_{c]} S \right) & = 0 \, , & k^a k^b A_{abc} & = 0 \, , & k^a \nabla_a S & = 0 \, .
\end{align}
 \item If $\Phi_{ab}$ is of Petrov type N, then
\begin{align}\label{eq-obstruction-N}
k^c \left( A_{abc} - g_{ca} \nabla_b S \right) & = 0 \, , & k^a A_{abc} & = 0 \, , & k_{[a} \nabla_{b]} S & = 0 \, .
\end{align}
\end{itemize}
\end{prop}

\begin{proof}
With reference to the Newman-Penrose formalism, and in a frame adapted to $\mcN$, we first note that
\begin{align*}
\mbox{Condition \eqref{eq-obstruction-II}} & & \Longleftrightarrow & & \left\{
\begin{aligned}
A _0 \equiv 0 \, , \\
A _1 + 3 \, D S \equiv 0 \, ,
\end{aligned} \right. \\
\mbox{Condition \eqref{eq-obstruction-III}} & & \Longleftrightarrow & & \left\{
\begin{aligned}
A _0 = A _1 \equiv 0 \, , \qquad D S = 0 \\
A _2 + 2 \, \delta S \equiv 0 \, .
\end{aligned} \right. \\
\mbox{Condition \eqref{eq-obstruction-N}} & & \Longleftrightarrow & & \left\{
\begin{aligned}
A _0 = A _1 = A_2 \equiv 0 \, , \qquad D S = \delta S = 0 \\
A _3 + \Delta S \equiv 0 \, .
\end{aligned} \right.
\end{align*}
The assumption that $\mcN$ is co-geodetic is simply $\kappa \equiv 0$ and $\rho \equiv 0$. We now deal with each case separately, referring to the Newman-Penrose equations given in Appendix \ref{sec-NP}.

\begin{itemize}
\item Assuming the type II condition, i.e. $\Phi _0 = \Phi _1 \equiv 0$, we have
\begin{align*}
\eqref{Cotton0000} : & & A _0 & \equiv 0 \, , \\
3 \times \eqref{Bianchi00} + \eqref{Cotton0001} : & & A _1 + 3 \, D S & \equiv 0 \, .
\end{align*}
\item Assuming the type III condition, i.e. $\Phi _0 = \Phi _1 = \Phi _2 \equiv 0$, we have
\begin{align*}
 \eqref{Bianchi00} : & & D S & =  0 \, , \\
 \eqref{Cotton0000} : & & A _0 & = 0 \, , \\
 \eqref{Cotton0001} : & & A _1 & = 0 \, , \\
2 \times \eqref{Bianchi01} + \eqref{Cotton0011} : & & A _2 + 2 \, \delta S & \equiv 0 \, .
\end{align*}
\item Assuming the type N condition, i.e. $\Phi _0 = \Phi _1 = \Phi _2 = \Phi _3 \equiv 0$, we have
\begin{align*}
\eqref{Bianchi00} : & & D S & = 0 \, , \\
\eqref{Bianchi01} : & & \delta S & = 0 \, , \\
\eqref{Cotton0000} : & & A _0 & = 0 \, , \\
\eqref{Cotton0001} : & & A _1 & = 0 \, , \\
\eqref{Cotton0011} : & & A _2 & = 0 \, , \\
\eqref{Bianchi11} + \eqref{Cotton0111} : & & A _3 + \Delta S & \equiv 0 \, .
\end{align*}
\end{itemize}
Comparison with the frame components \eqref{eq-A-components} completes the proof.
\end{proof}

\begin{rem}
Proposition \ref{prop-obstruction-GS} also applies to tracefree Ricci tensors of Petrov type D, in which case one has a pair of distinct multiple principal null structures as described in Remark \ref{rem-typeD}.
\end{rem}

\subsection{Algebraic speciality implies co-geodetic null structures}\label{sec-algsp-co-geo}
We are now in the position of formulating the Goldberg-Sachs theorems (Theorems \ref{thm-GS-II-gen}, \ref{thm-GS-III-gen} and \ref{thm-GS-N-gen}), along lines similar to Kundt \& Thompson \cite{Kundt1962} and Robinson \& Schild \cite{Robinson1963}. Note however that unlike the versions of these authors, the following theorems are \emph{not} conformally invariant.

\begin{thm}[Petrov type II]\label{thm-GS-II-gen}
Let $\mcN$ be a multiple principal null structure on $(\mcM,\bm{g})$. Assume that $\Phi_{ab}$ is of Petrov type II and does not degenerate further. Suppose further that, for any generator $k^a$ of $\mcN$,
\begin{align}\tag{\ref{eq-obstruction-II}}
k^a k^b \left( A_{abc} + 3 \, g_{bc} \nabla_a S \right) & = 0 \, .
\end{align}
Then $\mcN$ is co-geodetic.
\end{thm}

\begin{proof}
Assume that $\Phi_{ab}$ is of Petrov type II, i.e. $\Phi _0 = \Phi _1 = 0$ in an adapted frame. In this case, the Newman-Penrose equations give
\begin{align*}
\eqref{Cotton0000} : & & A _0 & = - 12 \, \kappa \Phi _2 \, , \\
3 \times \eqref{Bianchi00} + \eqref{Cotton0001} : & &  A _1 + 3 D S & =  
 - 12 \, \rho \Phi _2 
\, .
\end{align*}
The assumption \eqref{eq-obstruction-II} in an adapted frame tells us that the LHS of the above set of equations are precisely zero. Now, since $\Phi_{ab}$ does not degenerate further, $\Phi _2 \neq 0$, so we conclude $\kappa \equiv 0$ and $\rho \equiv 0$.
\end{proof}

\begin{thm}[Petrov type III]\label{thm-GS-III-gen}
Let $\mcN$ be a multiple principal null structure on $(\mcM,\bm{g})$. Assume that $\Phi_{ab}$ is of Petrov type III and does not degenerate further. Suppose further that, for any generator $k^a$ of $\mcN$,
\begin{align}\tag{\ref{eq-obstruction-III}}
k^a \left( A_{abc} - 2 \, g_{a[b} \nabla_{c]} S \right) & = 0 \, , & k^a k^b A_{abc} & = 0 \, , & k^a \nabla_a S & = 0 \, ,
\end{align}
Then $\mcN$ is co-geodetic.
\end{thm}

\begin{proof}
Assume that $\Phi_{ab}$ is of Petrov type III, i.e. $\Phi _0 = \Phi _1 = \Phi_2 = 0$ in an adapted frame. In this case, the Newman-Penrose equations give
\begin{align*}
\eqref{Bianchi00} : & & D S & =  2 \, \kappa \Phi _3 
\, , \\
\eqref{Cotton0000} : & & A _0 & = 0
\, , \\
\eqref{Cotton0001} : & & A _1 & = 
- 6 \, \kappa \Phi _3 \, , \\
2\times\eqref{Bianchi01}+\eqref{Cotton0011} : & & A _2 + 2 \, \delta S & =  - 8 \, \rho \Phi _3 
\, .
\end{align*}
The assumption \eqref{eq-obstruction-III} in an adapted frame tells us that the LHS of the above set of equations are precisely zero. Now, since $\Phi_{ab}$ does not degenerate further, $\Phi _3 \neq 0$, so we conclude $\kappa \equiv 0$ and $\rho \equiv 0$.
\end{proof}

\begin{thm}[Petrov type N]\label{thm-GS-N-gen}
Let $\mcN$ be a multiple principal null structure on $(\mcM,\bm{g})$. Assume that $\Phi_{ab}$ is of Petrov type N and does not degenerate further. Suppose further that, for any generator $k^a$ of $\mcN$,
\begin{align}\tag{\ref{eq-obstruction-N}}
k^c \left( A_{abc} - g_{ca} \nabla_b S \right) & = 0 \, , & k^a A_{abc} & = 0 \, , & k_{[a} \nabla_{b]} S & = 0 \, ,
\end{align}
Then $\mcN$ is co-geodetic.
\end{thm}

\begin{proof}
Assume that $\Phi_{ab}$ is of Petrov type N, i.e. $\Phi _0 = \Phi _1 = \Phi_2 = \Phi_3 = 0$ in an adapted frame. In this case, the Newman-Penrose equations give
\begin{align*}
\eqref{Bianchi00} : & & D S & =  0
\, , \\
\eqref{Bianchi01} : & & \delta S & = \kappa \Phi _4 
\, , \\
\eqref{Cotton0000} : & & A _0 & = 0
\, , \\
\eqref{Cotton0001} : & & A _1 & = 0 \, ,  \\
\eqref{Cotton0011} : & & A _2 & = - 2 \,  \kappa \Phi _4 
\, , \\
\eqref{Bianchi11} + \eqref{Cotton0111} :
& & \Delta S +  A _3 & = - 4 \, \rho \Phi _4 \, .
\end{align*}
The assumption \eqref{eq-obstruction-N} in an adapted frame tells us that the LHS of the above set of equations are precisely zero. Now, since $\Phi_{ab}$ does not degenerate further, $\Phi _4 \neq 0$, so we conclude $\kappa \equiv 0$ and $\rho \equiv 0$.
\end{proof}

\begin{thm}[Petrov type D]\label{thm-GS-D-gen}
Assume that $\Phi_{ab}$ is of Petrov type D with multiple principal null structures $\mcN$ and $\mcN'$ on $(\mcM,\bm{g})$, and does not degenerate further. Let $k^a$ and $\ell^a$ be any generators of $\mcN$ and $\mcN'$ respectively, and suppose further that
\begin{align}\label{eq-obstruction-D}
k^a k^b \left( A_{abc} + 3 \, g_{bc} \nabla_a S \right) & = 0 \, , & \ell^a \ell^b \left( A_{abc} + 3 \, g_{bc} \nabla_a S \right) & = 0 \, .
\end{align}
Then both $\mcN$ and $\mcN'$ are co-geodetic.
\end{thm}

\begin{proof}
Assume that $\Phi_{ab}$ is of Petrov type D, i.e. $\Phi _0 = \Phi _1 = 0 = \Phi_3 = \Phi_4$ in a frame adapted to both $\mcN$ and $\mcN'$ as explained in Remark \ref{rem-typeD}. In this case, we see that the additional constraints $\Phi_3 = \Phi_4 =0$ do not affect any of the argument of the proof of Theorem \ref{thm-GS-D-gen}, which impinges on the condition $\Phi_2 \neq 0$, and we can conclude $\mcN$ is geodetic.

To show that $\mcN'$ is integrable, we have to show that in an adapted frame, and with reference to the covariant derivative \eqref{eq-nablal} of $\ell^a$, the Newman-Penrose coefficients $\mu$ and $\nu$ should also be zero. But the Newman-Penrose equations give
\begin{align*}
\eqref{Cotton1111}: & & A _4 & = - 12 \, \nu \Phi _2 \, , \\
3 \times \eqref{Bianchi11} - \eqref{Cotton0111}: & & A _3 - 3 \Delta S & = - 12 \, \mu \Phi _2 \, .
\end{align*}
The assumption \eqref{eq-obstruction-D} in an adapted frame tells us that the LHS of the above set of equations are precisely zero. Now, since $\Phi_{ab}$ does not degenerate further, $\Phi _2 \neq 0$, so we conclude $\kappa \equiv 0$ and $\rho \equiv 0$.
\end{proof}

\subsection{Co-geodetic null structures implies algebraic speciality}\label{sec-co-geo-algsp}
We now state and prove the converse to Theorems \ref{thm-GS-II-gen}, \ref{thm-GS-III-gen} and \ref{thm-GS-N-gen}.
\begin{thm}\label{thm-GS-hard}
Let $\mcN$ be a co-geodetic null structure on $(\mcM,\bm{g})$. Suppose  that, for any generator $k^a$ of $\mcN$,
\begin{align}\tag{\ref{eq-obstruction-II}}
k^a k^b \left( A_{abc} + 3 \, g_{bc} \nabla_a S \right) & = 0 \, .
\end{align}
Then $\Phi_{ab}$ is algebraically special with $\mcN$ as multiple principal null structure.
\end{thm}

\begin{proof}
As always we work in an adapted frame and use the Newman-Penrose formalism, in which $\kappa = \rho \equiv 0$ means that $\mcN$ is co-geodetic. Then, assuming $\kappa = \rho \equiv 0$, we know by Proposition \ref{prop-integrability-cond} that $\Phi _0 \equiv 0$. In this case, we have the following components of the Bianchi identity
\begin{align}
D \Phi _2 - 2 \, \delta \Phi _1 - D S & =  \left( 2 \, \pi - 4 \, \tau - 4 \, \alpha \right) \Phi _1 
 \, , \tag{\ref{Bianchi00}}
\end{align}
and of the Cotton tensor
\begin{align}
- 4 \, D \Phi _1 & = A _0 
- 8 \, \epsilon \Phi _1 \, , \tag{\ref{Cotton0000}} \\
2 \, \delta \Phi _1 - 3 D \Phi _2 & = A _1 
+ \left( 4 \, \alpha - 4 \, \tau - 6 \, \pi \right) \Phi _1
\, , \tag{\ref{Cotton0001}}
\end{align}
We proceed by steps:
\begin{itemize}
 \item first, computing $3 \times \eqref{Bianchi00} +  \eqref{Cotton0001}$ yields
\begin{align}\label{eq-step1}
 - 4 \delta \Phi _1 & = \left( A _1 + 3 \, D S \right) - 8 \left( 2 \tau + \alpha \right) \Phi _1 \, ;
\end{align}
 \item then $\delta \eqref{Cotton0000} - D \eqref{eq-step1}$ gives
 \begin{align}\label{eq-step2}
4 \, [ D , \delta ] \Phi _1 & = \delta \, A _0 - D \left( A _1 + 3 D \, S \right) - 8 \left( \delta \epsilon - D \alpha - 2\, D \tau \right) \Phi _1 - 8 \, \epsilon \, \delta \Phi _1 + 8 \, \left( \alpha +2 \tau \right) D \Phi _1 \, ;
 \end{align}
 \item at this stage, we can substitute the commutation relation
\begin{align}
 [ D , \delta ] & =  \left( \pi - 2 \, \alpha \right) D \, , \tag{\ref{commutator0001}}
\end{align}
into the LHS of \eqref{eq-step2}, and the following components of the Ricci identity
\begin{align}
D \tau & =
 - 2 \, \Phi _1 \, , \tag{\ref{Ricci001100}} \\
D \alpha - \delta \epsilon & =
- 2 \, \epsilon \alpha 
+ \pi \epsilon 
- \Phi _1 \, , \tag{\ref{Ricci000101}}
\end{align}
together with \eqref{eq-step1} and \eqref{Cotton0000} into the RHS of \eqref{eq-step2}, to get
 \begin{multline}\label{eq-step3}
4 \,  \left( \pi - 2 \, \alpha \right) D \Phi _1 = \delta \, A _0 - D \left( A _1 + 3 D \, S \right) + 2 \, \epsilon \,  \left( A _1 + 3 \, D S \right) - 2 \, \left( \alpha +2 \tau \right) A _0  \\
- 8 \left( 2 \, \epsilon \alpha - \pi \epsilon + 5 \, \Phi _1 \right) \Phi _1 \cancel{ - 16 \, \epsilon \,   \left( 2 \tau + \alpha \right) \Phi _1 + 16 \, \left( \alpha +2 \tau \right) \epsilon \Phi _1 } \, ;
 \end{multline}
 \item subsituting \eqref{Cotton0000} into the RHS of \eqref{eq-step3} and expanding yields
 \begin{align}\label{eq-step4}
 40 \left( \Phi _1 \right)^2 & =  \delta \, A _0 - D \left( A _1 + 3 D \, S \right) + 2 \, \epsilon \,  \left( A _1 + 3 \, D S \right) + \left( \pi - 4 \tau - 4 \alpha \right) A _0 \, ;
 \end{align}
 \item finally, by condition \eqref{eq-obstruction-II}, the RHS of \eqref{eq-step4} vanishes identically and we conclude $\Phi _1 \equiv  0$.
\end{itemize}
In summary, $\kappa = \rho \equiv 0$ implies $\Phi _0 = \Phi _1 \equiv  0$, i.e. co-geodetic $\mcN$ implies algebraic speciality of $\Phi_{ab}$ with $\mcN$ multiple principal null structure.
\end{proof}

\subsection{Topological massive gravity}\label{sec-GS-TMG}
Next, we consider the equations governing topological massive gravity. These are none other than Einstein's equations with cosmological constant $\Lambda$ in which the energy-momentum tensor is proportional to the Hodge-dual of the Cotton tensor, i.e.
\begin{align}\label{eq-TMG1}
R_{ab} - \frac{1}{2} R \, g_{ab} + \Lambda \, g_{ab} & = \frac{1}{m} (*A)_{ab} \, .
\end{align}
Here, $m \neq 0$ is a parameter of topological massive gravity theory.
Substituting the expression for the tracefree Ricci tensor and tracing yield the expressions
\begin{align}
\Phi_{ab} & = \frac{1}{m} (*A)_{ab} \, , \label{eq-TMG2a} \\
R & = 6 \, \Lambda = \mbox{constant} \, , \label{eq-TMG2b}
\end{align}
equivalent to \eqref{eq-TMG1}.

\begin{rem}
It is in fact sufficient to consider only \eqref{eq-TMG2a} since \eqref{eq-TMG2b} follows from it. To see this, we note that $\nabla^a (*A)_{ab} = 0$ which follows from \eqref{eq-divCotton}. So, by \eqref{eq-TMG2a}, $\nabla^a \Phi_{ab}  = 0$. Now, the Bianchi identity \eqref{eq-Bianchi-contracted} gives $\nabla_a R = 0$, i.e. \eqref{eq-TMG2b}.
\end{rem}

In an adapted frame, and with reference to \eqref{eq-NP-Cotton}, equations \eqref{eq-TMG2a} read as
\begin{align}\label{eq-TMG2a-comp}
\Phi_0 & = -\frac{\ii^q}{2\sqrt{2}m} A_0 \, , & \Phi_1 & = -\frac{\ii^q}{2\sqrt{2}m} A_1 \, , & \Phi_2 & = -\frac{\ii^q}{2\sqrt{2}m} A_2 \, , & \Phi_3 & = -\frac{\ii^q}{2\sqrt{2}m} A_3 \, , & \Phi_4 & = -\frac{\ii^q}{2\sqrt{2}m} A_4 \, ,
\end{align}
where $q=0$ in Euclidean signature, and $q=1$ in Lorentzian signature.

\begin{lem}
Suppose $(\mcM,\bm{g})$ is a solution of the topological massive gravity equations \eqref{eq-TMG2a} and \eqref{eq-TMG2b}.
\begin{itemize}
\item If the tracefree Ricci tensor is of Petrov type $II$, then
\begin{align*}\tag{\ref{eq-obstruction-II}}
k^a k^b \left( A_{abc} + 3 \, g_{bc} \nabla_a S \right) & = 0 \, .
\end{align*}
\item If the tracefree Ricci tensor is of Petrov type $III$, then
\begin{align*}\tag{\ref{eq-obstruction-III}}
k^a \left( A_{abc} - 2 \, g_{a[b} \nabla_{c]} S \right) & = 0 \, , & k^a k^b A_{abc} & = 0 \, , & k^a \nabla_a S & = 0 \, .
\end{align*}
\item If the tracefree Ricci tensor is of Petrov type $N$, then
\begin{align*}\tag{\ref{eq-obstruction-N}}
k^c \left( A_{abc} - g_{ca} \nabla_b S \right) & = 0 \, , & k^a A_{abc} & = 0 \, , & k_{[a} \nabla_{b]} S & = 0 \, .
\end{align*}
\end{itemize}
\end{lem}

\begin{proof}
We first note that $S=\frac{1}{2}\Lambda$ is constant by virtue of the topological massive gravity equations \eqref{eq-TMG2b}. Consequently, equations \eqref{eq-obstruction-II}, \eqref{eq-obstruction-III} and \eqref{eq-obstruction-N}, which we need to assert, are reduced to algebraic conditions on the Cotton tensor. More precisely, with respect to an adapted frame, we must now show that
\begin{align*}
\mbox{Petrov type II:} \qquad \Phi_0 = \Phi_1 & \equiv 0 & & \Rightarrow & A_0 = A_1 & = 0 \, , \\
\mbox{Petrov type III:} \qquad \Phi_0 = \Phi_1 = \Phi_2 & \equiv 0  & & \Rightarrow & A_0 = A_1 = A_2 & = 0 \, , \\
\mbox{Petrov type N:} \qquad \Phi_0 = \Phi_1 = \Phi_2 = \Phi_3 & \equiv 0  & & \Rightarrow & A_0 = A_1 = A_2 = A_3 & = 0 \, . 
\end{align*}
But the veracity of these statements follows from the topological massive gravity equations \eqref{eq-TMG2a}, which are \eqref{eq-TMG2a-comp} in an adapted frame.
\end{proof}

Combining Theorems \ref{thm-GS-II-gen}, \ref{thm-GS-III-gen}, \ref{thm-GS-N-gen}, \ref{thm-GS-D-gen} and \ref{thm-GS-hard} leads to our main result.
\begin{thm}\label{thm-GS-TMG}
Let $(\mcM, \bm{g})$ be an oriented three-dimensional (pseudo-)Riemannian manifold that is a solution of the topological massive gravity equations, and assume that the Petrov type of the tracefree Ricci tensor $\Phi_{ab}$ does not change in an open subset of $\mcM$. Then $\Phi_{ab}$ is algebraically special if and only if $(\mcM, \bm{g})$ admits a co-geodetic null structure.
\end{thm}

\subsection{Real versions}
All the theorems given in sections \ref{sec-obst}, \ref{sec-algsp-co-geo}, \ref{sec-co-geo-algsp} and \ref{sec-GS-TMG} can easily be adapted to the case where the metric is real-valued. The crucial points here are that
\begin{itemize}
\item each of the algebraically special Petrov types of the tracefree Ricci tensor, as given in sections \ref{sec-Petrov-types-Euc} and \ref{sec-Petrov-types-Lor}, singles out multiple principal null structure of a particular real index, and
\item the real index $r$ of this null structure yields a particular \emph{real} geometric interpretation as given in section \ref{sec-real-metrics}, i.e. a congruence of null curves when $r=1$, or a congruence of timelike curves when $r=0$.
\end{itemize}

We shall only give real versions of Theorem \ref{thm-GS-TMG} in the context of the topological massive gravity equations. 

\begin{thm}[Lorentzian Goldberg-Sachs theorem for Topological Massive Gravity]\label{thm-Lor-GS-TMG}
Let $(\mcM,\bm{g})$ be an oriented  three-dimensional Lorentzian manifold that is a solution of the topological massive gravity equations. Then
\begin{itemize}
\item $(\mcM,\bm{g})$ admits a divergence-free congruence of null geodesics (i.e. is a Kundt spacetime) if and only if its tracefree Ricci tensor is of Petrov type II, II${}_r$, D${}_r$, III${}_r$ or N${}_r$;
\item $(\mcM,\bm{g})$ admits two distinct divergence-free congruences of null geodesics if and only if its tracefree Ricci tensor is of Petrov type D${}_r$;
\item $(\mcM,\bm{g})$  admits a shear-free congruence of timelike geodesics if and only if its tracefree Ricci tensor is of Petrov type D.
\end{itemize}
\end{thm}

In fact, parts of Theorem \ref{thm-Lor-GS-TMG} were proved in reference in \cite{Chow2010}: namely, that every Kundt spacetime that is solution of the topological massive gravity equations must be algebraically special. By Theorem \ref{thm-Lor-GS-TMG}, this exhausts all solutions of Petrov types II, II${}_r$, D${}_r$, III${}_r$ or N${}_r$. All Petrov type D solutions of the topological massive gravity equations are also given in reference \cite{Chow2010a}. Therefore, Theorem \ref{thm-Lor-GS-TMG} tells us that these are the only possible algebraically special solutions of the topological massive gravity equations.

For the sake of completeness, we state the Riemannian version of Theorem \ref{thm-Riem-GS-TMG}.

\begin{thm}[Riemannian Goldberg-Sachs theorem for Topological Massive Gravity]\label{thm-Riem-GS-TMG}
Let $(\mcM,\bm{g})$ be an oriented three-dimensional Riemannian manifold that is a solution of the topological massive gravity equations. Then $(\mcM,\bm{g})$  admits a shear-free congruence of geodesics if and only if its tracefree Ricci tensor is algebraically special, i.e. of Petrov type D.
\end{thm}

\appendix
\section{Spinor calculus in three dimensions}\label{app-spinor-calculus}
Let $(\mcM,\bm{g})$ be a three-dimensional (pseudo-)Riemannian manifold, which we shall also assume to be oriented and equipped with a spin structure. To make the discussion signature-independent, we shall as before complexify both the tangent bundle $\Tgt \mcM$ and $\bm{g}$, and shall not assume the existence of any particular reality structure on $\Tgt^\C \mcM$ preserving $\bm{g}$. With these considerations, the spinor bundle $\mcS$ over $\mcM$ is a complex rank-$2$ vector bundle, whose sections will carry upstairs upper-case Roman indices, e.g. $\alpha^A \in \Gamma (\mcS)$. Similarly, sections of the dual spinor bundle $\mcS^*$ will carry downstairs upper-case Roman indices, e.g. $\beta_A \in \Gamma (\mcS^*)$. The bundles $\mcS$ and $\mcS^*$ are equipped with non-degenerate skew-symmetric $2$-spinors $\varepsilon_{AB}$ and $\varepsilon^{AB}$ respectively, which we shall choose to satisfy the normalisation condition
\begin{align*}
	\varepsilon_{AC} \varepsilon^{BC} & = \delta \ind*{_A^B} \, ,
\end{align*}
where $\delta_A^B$ is the identity on $\mcS$. These bilinear forms establish an isomorphism between $\mcS$ and its dual $\mcS^*$, and we shall then raise and lower indices on spinors and dual spinors according to the convention
\begin{align*}
\alpha_A & = \alpha^B \varepsilon_{BA} \, , & \beta^A & = \varepsilon ^{AB} \beta_B\, .
\end{align*}
This spinor calculus is almost identical to the two-spinor calculus in four dimensions, except for the absence of chirality (i.e. of `primed' spinor indices).

We can consider the tensor product of any number copies of $\mcS$. Since the fibers of $\mcS$ are two-dimensional, any skew-symnetric $2$-spinor must be pure trace, i.e. $\phi_{[AB]} = \frac{1}{2} \varepsilon_{AB} \phi \ind{_C^C}$. In particular, there is a natural isomorphism between $\odot^2 \mcS$ and $\Tgt^\C \mcM$, and, by Hodge duality, $\wedge^2 \Tgt^\C \mcM$. This means that vector fields can be represented by a symmetric $2$-spinor, i.e.
\begin{align*}
V^a = V^{AB} \, .
\end{align*}
where $V^{AB} = V^{(AB)}$. Here, we are employing the abstract index notation of \cite{Penrose1984}. For those uncomfortable with this approach, one can introduce $\gamma$-matrices $\gamma \ind{_a^{AB}}$ to convert spinorial indices into tensorial ones and vice versa, i.e.
\begin{align*}
V^{AB} & = V^a \gamma \ind{_a^{AB}} \, , & V^a & = \gamma \ind{^a_{AB}} V^{AB} \, .
\end{align*}
These $\gamma$-matrices satisfy $2 \, \gamma \ind{_{(a}_A^C} \gamma \ind{_{b)}_{BC}} =  g_{ab} \varepsilon_{AB}$.

As in four dimensions, the metric $g_{ab}$ can be reinterpreted as the outer product of two copies $\varepsilon_{AB}$, which we find to be
\begin{align}\label{eq-metric-spinor}
g_{ab} & = g_{ABCD} = g_{(AB)(CD)} = - \varepsilon _{A(C} \varepsilon _{D)B} \, .
\end{align}
It follows that the norm of any vector field $V^a$ at any point equals the Pfaffian of its corresponding spinor $V^{AB}$, i.e.
\begin{align*}
V^a V_a & = \varepsilon _{AC} \varepsilon _{BD} V^{AB} V^{CD} \, .
\end{align*}
Hence, a non-zero vector field $V^a$ is null if and only if its corresponding spinor $V^{AB}$ has vanishing Pfaffian. Hence, $V^{AB}$ must be of rank $1$, and we can write $V^{AB} = \alpha^A \beta^B$ for some spinors $\alpha^A$ and $\beta^A$. In fact, using \eqref{eq-metric-spinor} once more, we see that $\alpha^A \beta_A = 0$, and so $\beta^A$ must be proportional to $\alpha^A$. The constant of proportionality can be absorbed by the spinor so that

\begin{lem}
Any null vector field $k^a$ can be written in the form
\begin{align*}
k^a = k^{AB} & = \xi^A \xi^B \, ,
\end{align*}
for some spinor field $\xi^A$.

In particular, there is a one-to-one correspondence between null line subbundle of $\Tgt^\C \mcM$ and lines of spinor fields.
\end{lem}

\paragraph{Decomposition of a $2$-form}
Any $2$-form on $\mcM$ can be expressed as
\begin{align}\label{eq-2form-spinor}
\phi_{ab} = \phi_{[ab]} = \phi _{ABCD} = \phi _{(AB)(CD)} = \phi_{ABCD} = 2 \, \varepsilon_{(A|(C} \phi_{D)|B)}
\end{align}
where $\phi_{AB} = \phi_{(AB)} = \frac{1}{2} \phi \ind{_{ACB}^C}$.

\paragraph{Curvature spinors}
The decomposition rule \eqref{eq-2form-spinor} allows us to write the Riemann tensor as
\begin{align*}
R \ind{_{abcd}} & = R \ind{_{(AB)(CD)(EF)(GH)}}
 = 4 \, \varepsilon_{(A|(C} X_{D)|B)(E|(G} \varepsilon_{H)|F)} \, .
\end{align*}
where $X_{ABCD} = X_{(AB((CD)}$ satisfies $X_{ABCD} = X_{CDAB}$.
Writing
\begin{align*}
2 \, \Phi \ind{_{ABCD}} & = \Phi_{ab} \, , &  R & = R \ind{_a^a} =: 12 \, S \, ,
\end{align*}
for the tracefree Ricci tensor and the Ricci scalar respectively. Here the factor $2$ preceding $\Phi_{ABCD}$ has been added for later convenience. We can show
\begin{align*}
X \ind{_{ABCD}} & = \Phi \ind{_{ABCD}} + S \, \varepsilon \ind{_{A(C}} \varepsilon \ind{_{D)B}} \, .
\end{align*}
where $\Phi \ind{_{ABCD}} = \Phi \ind{_{(ABCD)}}$ and $X \ind{_{ACB}^C} = 3 \, S \, \varepsilon_{AB}$ and $X \ind{_{AB}^{AB}} = 6 \, S$.
Contracting yields
\begin{align*}
R_{ab} & = 2 X _{(A|(CD)|B)} - 2 \, \varepsilon_{(A|(C} X \ind{_{D)}^E_{E|B)}}  \, .
\end{align*}

\paragraph{Spinor geometry}
Applying the decomposition \eqref{eq-2form-spinor} to the commutator yields
\begin{align}\label{eq-commutator}
2 \, \nabla_{[a} \nabla_{b]} & = \nabla_{AB} \nabla_{CD} - \nabla_{CD} \nabla_{AB}
= 2 \, \varepsilon _{(A|(C} \Box{_{D)|B)}}
\end{align}
where $\Box \ind{_{AB}} := \nabla \ind{_{C (A}} \nabla \ind{_{B)}^C}$, which on any spinor $\alpha^A$, acts as
\begin{align}\label{eq-commutator-spinor}
\Box_{AB} \alpha^C & = - X \ind{_{ABD}^C} \alpha^D \, ,
\end{align}
where $X \ind{_{ABCD}}$ is the curvature spinor.

The contracted Bianchi identitiy \eqref{eq-Bianchi-contracted} in spinorial form reads
\begin{align}\label{eq-Bianchi-spinor}
\nabla^{CD} \Phi \ind{_{CDAB}} - \nabla _{CD} S & = 0 \, ,
\end{align}
while the Cotton tensor \eqref{eq-Cotton} or \eqref{eq-Cotton-dual} reads
\begin{align}\label{eq-Cotton-spinor}
A_{ABCD} & = 4 \, \nabla \ind{_{(A}^E} \Phi \ind{_{BCD)E}} \, .
\end{align}
As a matter of interest, we record the topological massive gravity equations \eqref{eq-TMG2a} and \eqref{eq-TMG2b} in spinorial form
\begin{align*}
\Phi \ind{_{ABCD}} & = - \frac{\ii^q}{2 \sqrt{2}m} A \ind{_{ABCD}} \, , & 2 \, S & = \Lambda = \mbox{constant} \, ,
\end{align*}
where $q=0$ in Euclidean signature and $q=1$ in Lorentzian signature, and $m$ is a constant.

\paragraph{A useful formula}
We can convenient split the image of $\Phi_{ABCD}$ under the Dirac operator into irreducibles, in the sense of
\begin{align*}
\nabla \ind{_A^E} \Phi \ind{_{BCDE}} & = \frac{1}{4} \left( - 3 \, \varepsilon \ind{_{A(B}} \nabla \ind{^{EF}} \Phi \ind{_{CD)EF}} + 4 \, \nabla \ind{_{(A}^E} \Phi \ind{_{BCD)E}} \right) \, .
\end{align*}
Now, by the Bianchi identity \eqref{eq-Bianchi-spinor} and the definition of the Cotton `spinor' \eqref{eq-Cotton-spinor}, we obtain the useful identity
\begin{align}\label{eq-Bianchi+Cotton}
\nabla \ind{_A^E} \Phi \ind{_{BCDE}} & = \frac{1}{4} \left( - 3 \, \varepsilon \ind{_{A(B}} \nabla \ind{_{CD)}} S + A_{ABCD} \right) \, .
\end{align}

\paragraph{Reality conditions}
When $\bm{g}$ has signature $(3,0)$ the spin group  is isomorphic to $\SU(2)$, while when $\bm{g}$ has signature $(2,1)$, the spin group is isomorphic to $\SL(2,\R)$. Both are real forms of the complex Lie group $\SL(2,\C)$.

\section{A Newman-Penrose formalism in three dimensions}\label{sec-NP}
Our starting point will be a spin dyad $( o^A , \iota^A )$ normalised to $o_A \iota^A =1$. We shall adopt the convention that
\begin{align*}
o^A & = \delta \ind*{_0^A} \, , & \iota^A & = \delta \ind*{_1^A} \, , & \iota_A & = - \delta \ind*{_A^0} \, , & o_A & = \delta \ind*{_A^1} \, ,
\end{align*}
where we think of $\delta_A^B$ as a Kronecker delta.
Thus, to take the components of a spinor $S \ind{_{ABC}}$, say, with respect to this dyad, we shall write
\begin{align*}
S_{ABC} o^A o^B o^C & = S_{000} \, , & S_{ABC} o^A \iota^B o^C & = S_{010} \, , & \ldots
\end{align*}
and so on. The spin-invariant bilinear form then takes the form $\varepsilon_{AB} = 2 \, o_{[A} \iota_{B]}$.
With respect to the spin dyad, the components of $\varepsilon _{AB}$ and its inverse $\varepsilon ^{AB}$ are given by
\begin{align*}
\varepsilon_{01} & = - \varepsilon_{10} = 1 \, , & \varepsilon^{01} & = - \varepsilon^{10} = 1 \, .
\end{align*}
This normalised spin dyad determines a null triad
\begin{align*}
k^a & := o^A o^B \, , & \ell^a & := \iota^A \iota^B \, , & n^a & := o^{(A} \iota^{B)} \, , 
\end{align*}
so that $k^a \ell_a = 1$ and $n^a n_b = - \frac{1}{2}$, and all other contractions vanish. The metric then takes the form
\begin{align}\label{eq-metric}
g_{ab} & = 2 \, k_{(a} \ell_{b)} - 2 \, n_a n_b \, .
\end{align}

\paragraph{Spin coefficients}
As before, we let $\nabla_{AB}$ denote the Levi-Civita connection preserving $g_{ab}$, and by extension its lift to the spinor bundle. We introduce a connection $\partial \ind{_{AB}}$, which preserves $g_{ab}$ together with the spin dyad $\{ o^A , \iota^A \}$. Then the difference between $\nabla_{AB}$ and $\partial_{AB}$ on any spinor $\xi^A$ will be given by
\begin{align*}
\nabla \ind{_{AB}} \xi^C & = \partial \ind{_{AB}} \xi^C + \gamma \ind{_{ABD}^C} \xi^D \, ,
\end{align*}
for some spinor $\gamma _{ABCD} = \gamma _{(AB)(CD)}$. This spinor can then be interpreted as the connection $1$-form of the Levi-Civita connection.

We define the following differential operators
\begin{align*}
\begin{pmatrix}
D \\ \Delta \\ \delta
\end{pmatrix}
& :=
\begin{pmatrix}
o^A o^B \nabla_{AB}  \\ \iota^A \iota^B \nabla_{AB} \\ o^A \iota^B \nabla_{AB} 
\end{pmatrix}
=
\begin{pmatrix}
k^a \nabla_a \\ \ell^a \nabla_a \\ n^a \nabla_a 
\end{pmatrix} \, .
\end{align*}
Then, we can express the components of the connection $1$-form $\gamma \ind{_{ABC}^D}$
\begin{align*}
\begin{pmatrix}
\kappa & \rho & \tau \\
\epsilon & \alpha & \gamma \\
\pi & \mu & \nu 
\end{pmatrix} :=
\begin{pmatrix}
\gamma \ind{_{0000}} & \gamma \ind{_{0100}} & \gamma \ind{_{1100}} \\
\gamma \ind{_{0001}} & \gamma \ind{_{0101}} & \gamma \ind{_{1101}} \\
\gamma \ind{_{0011}} & \gamma \ind{_{0111}} & \gamma \ind{_{1111}} 
\end{pmatrix}
& =
\begin{pmatrix}
o^B D o_B & o^B \delta o_B & o^B \Delta  o_B \\
\iota^B D o_B & \iota^B \delta o_B & \iota^B \Delta o_B \\
\iota^B D \iota_B & \iota^B \delta \iota_B & \iota^B \Delta \iota_B 
\end{pmatrix} \\
& =
\begin{pmatrix}
n^b D k_b & n^b \delta k_b & n^b \Delta  k_b \\
\frac{1}{2} \ell^b D k_b & \frac{1}{2} \ell^b \delta k_b & \frac{1}{2} \ell^b \Delta k_b \\
- n^b D \ell_b & - n^b \delta \ell_b & - n^b \Delta  \ell_b 
\end{pmatrix} \, .
\end{align*}
Expanding the covariant derivatives of $k^a$, $\ell^a$ and $n^a$ in terms of the spin coefficients yield
\begin{align}
\nabla_a k^b & = 2 \, \gamma \, k_a k^b + 2 \, \epsilon \, \ell_a k^b - 4 \, \alpha \, n_a k^b - 4 \, \tau \, k_a n^b - 4 \, \kappa \, \ell_a n^b + 8 \, \rho \, n_a n^b \, , \tag{\ref{eq-nablak}}\\
\nabla_a \ell^b & = - 2 \, \epsilon \, \ell_a \ell^b - 2 \, \gamma \, k_a \ell^b + 4 \, \alpha \, n_a \ell^b + 4 \, \pi \, \ell_a n^b + 4 \, \nu \, k_a n^b - 8 \, \mu \, n_a n^b \, , \tag{\ref{eq-nablal}}\\
\nabla_a n^b & = - 2 \, \kappa \, \ell_a \ell^b + 2 \, \nu k_a k^b + 2 \, \pi \, \ell_a k^b - 2 \, \tau \, k_a \ell^b + 4 \, \rho \, n_a \ell^b - 4 \, \mu \, n_a k^b \, . \tag{\ref{eq-nablau}}
\end{align}

\paragraph{Curvature coefficients}
Similarly, the components of the tracefree Ricci tensor are given by
\begin{align*}
\begin{pmatrix}
 \Phi_0 \\ \Phi_1 \\ \Phi_2 \\ \Phi_3 \\ \Phi_4
\end{pmatrix}
:=
\begin{pmatrix}
 \Phi_{0000} \\ \Phi_{0001} \\ \Phi_{0011} \\ \Phi_{0111} \\ \Phi_{1111} 
\end{pmatrix}
& = 
\begin{pmatrix}
 \Phi _{ABCD} o^A o^B o^C o^D \\ \Phi _{ABCD} o^A o^B o^C \iota^D \\ \Phi _{ABCD} o^A o^B \iota^C \iota^D \\ \Phi _{ABCD} o^A \iota^B \iota^C \iota^D \\ \Phi_{ABCD} \iota^A \iota^B \iota^C \iota^D 
\end{pmatrix} = 
\begin{pmatrix}
\frac{1}{2} \Phi _{ab} k^a k^b \\ \frac{1}{2} \Phi _{ab} k^a n^b \\ \frac{1}{2} \Phi _{ab} k^a \ell^a = \frac{1}{2} \Phi _{ab} n^a n^b \\ \frac{1}{2} \Phi _{ab} \ell^a n^b \\ \frac{1}{2} \Phi _{ab} \ell^a \ell^b 
\end{pmatrix} \, ,
\end{align*}
while those of the Cotton tensor by
\begin{align}\label{eq-NP-Cotton}
\begin{pmatrix}
 A_0 \\ A_1 \\ A_2 \\ A_3 \\ A_4
\end{pmatrix}
:=
\begin{pmatrix}
 A_{0000} \\ A_{0001} \\ A_{0011} \\ A_{0111} \\ A_{1111} 
\end{pmatrix}
& = 
\begin{pmatrix}
 A _{ABCD} o^A o^B o^C o^D \\ A _{ABCD} o^A o^B o^C \iota^D \\ A _{ABCD} o^A o^B \iota^C \iota^D \\ A _{ABCD} o^A \iota^B \iota^C \iota^D \\ A_{ABCD} \iota^A \iota^B \iota^C \iota^D 
\end{pmatrix} = 
\begin{pmatrix}
2 \, A _{abc} k^a k^b n^c \\ A _{abc} k^a k^b \ell^c \\ 2 \, A _{abc} k^a n^b \ell^c \\ A _{abc} \ell^a k^b \ell^c \\  2 \, A _{abc} \ell^a n^b \ell^c 
\end{pmatrix} = 
\begin{pmatrix}
- \sqrt{2} (-\ii)^q \, (*A) _{ab} k^a k^b \\ - \sqrt{2} (-\ii)^q \, (*A) _{ab} k^a n^b \\ - \sqrt{2} (-\ii)^q \, (*A) _{ab} k^a \ell^b \\ - \sqrt{2} (-\ii)^q \, (*A) _{ab} \ell^a n^b \\  - \sqrt{2} (-\ii)^q \, (*A) _{ab} \ell^a \ell^b 
\end{pmatrix} \, ,
\end{align}
where $q=0$ in Euclidean signature and $q=1$ in Lorentzian signature, and we have assumed that the volume form is given by
\begin{align*}
e_{abc} & = \ii^q 6 \sqrt{2} k_{[a} \ell_b n_{c]} \, .
\end{align*}

\paragraph{Commutation relations}
The commutator of the Levi-Civita connection
\begin{align*}
 [ \nabla_{A B} , \nabla_{C D} ] & = 2 \, \gamma \ind{_{A B (C}^E} \nabla \ind{_{D) E}} - 2 \, \gamma \ind{_{C D (A}^E} \nabla \ind{_{B) E}}
\end{align*}
has components given by
\begin{align}
 [ D , \Delta ] & = 2 \, \left( \pi + \tau \right) \delta - 2 \, \gamma D - 2 \, \epsilon \Delta \, , \label{commutator0011} \\
 [ D , \delta ] & = 2 \, \rho \delta 
 + \left( \pi - 2 \, \alpha \right) D - \kappa \Delta \, , \label{commutator0001} \\
 [ \Delta , \delta ] & = - 2 \, \mu \delta
 + \nu D + \left( - \tau + 2 \, \alpha \right) \Delta  \, . \label{commutator1101} 
\end{align}

\paragraph{Ricci identity}
The Ricci identity \eqref{eq-commutator} together with \eqref{eq-commutator-spinor} can be re-expressed as
\begin{multline*}
\partial \ind{_{AB}} \gamma \ind{_{CDE}^{F}} - \partial \ind{_{CD}} \gamma \ind{_{ABE}^{F}} = \\
\gamma \ind{_{ABE}^{G}} \gamma \ind{_{CDG}^{F}} - \gamma \ind{_{CDE}^{G}} \gamma \ind{_{ABG}^{F}} - \gamma \ind{_{CDA}^{G}} \gamma \ind{_{GBE}^{F}} + \gamma \ind{_{ABC}^{G}} \gamma \ind{_{GDE}^{F}} - \gamma \ind{_{CDB}^{G}} \gamma \ind{_{GAE}^{F}} + \gamma \ind{_{ABD}^{G}} \gamma \ind{_{GCE}^{F}} \\
- \frac{1}{2} \left( \varepsilon \ind{_{AC}} \Phi \ind{_{DBE}^F} + \varepsilon \ind{_{AD}} \Phi \ind{_{CBE}^F} + \varepsilon \ind{_{BC}} \Phi \ind{_{DAE}^F} + \varepsilon \ind{_{BD}} \Phi \ind{_{CAE}^F} \right) \\
- \frac{1}{4} S \left( \varepsilon \ind{_{AC}} \varepsilon \ind{_{ED}} \varepsilon \ind{_B^F} + \varepsilon \ind{_{AC}} \varepsilon \ind{_{EB}} \varepsilon \ind{_D^F} + \varepsilon \ind{_{AD}} \varepsilon \ind{_{EC}} \varepsilon \ind{_B^F} + \varepsilon \ind{_{AD}} \varepsilon \ind{_{EB}} \varepsilon \ind{_C^F} \right. \\
\left. + \varepsilon \ind{_{BC}} \varepsilon \ind{_{ED}} \varepsilon \ind{_A^F} + \varepsilon \ind{_{BC}} \varepsilon \ind{_{EA}} \varepsilon \ind{_D^F} + \varepsilon \ind{_{BD}} \varepsilon \ind{_{EC}} \varepsilon \ind{_A^F} + \varepsilon \ind{_{BD}} \varepsilon \ind{_{EA}} \varepsilon \ind{_C^F} \right) \, .
\end{multline*}
Taking the various components with respect to the spin dyad $\{ o^A, \iota^A \}$ yields
\begin{align}
D \tau - \Delta \kappa & =
- 4 \, \gamma \kappa 
+ 2 \, \pi \rho + 2 \, \tau \rho 
- 2 \, \Phi _1 \, , \label{Ricci001100} \\
D \rho - \delta \kappa & =
2 \, \epsilon \rho - 4 \, \alpha \kappa 
+ 2 \, \rho \rho + \pi \kappa 
- \kappa \tau 
- \Phi _0 \, , \label{Ricci000100} \\
\Delta \rho - \delta \tau & =
2 \, \gamma \rho 
 - 2 \, \mu \rho + \nu \kappa 
- \tau \tau 
+ \Phi _2  - S \, , \label{Ricci110100} \\
D \gamma - \Delta \epsilon & =
- 4 \, \epsilon \gamma 
- \kappa \nu + \tau \pi 
+ 2 \, \pi \alpha 
+ 2 \, \tau \alpha 
- 2 \, \Phi _2 - S \, , \label{Ricci001101} \\
D \alpha - \delta \epsilon & =
- 2 \, \epsilon \alpha 
- \kappa \mu + \rho \pi 
+ \pi \epsilon 
+ 2 \, \rho \alpha - \kappa \gamma 
- \Phi _1 \, , \label{Ricci000101} \\
\Delta \alpha - \delta \gamma & =
2 \, \gamma \alpha 
- \tau \mu + \rho \nu 
- 2 \, \mu \alpha + \nu \epsilon 
- \tau \gamma 
+ \Phi _3 \, , \label{Ricci110101} \\
D \nu - \Delta \pi & =
- 4 \, \nu \epsilon 
 + 2 \, \pi \mu 
 + 2 \, \tau \mu 
- 2 \, \Phi _3 \, , \label{Ricci001111} \\
D \mu - \delta \pi & =
- 2 \, \mu \epsilon 
+ \pi \pi 
+ 2 \, \rho \mu - \kappa \nu 
- \Phi _2 + S \, , \label{Ricci000111} \\
\Delta \mu - \delta \nu & =
4 \, \nu \alpha - 2 \, \mu \gamma 
- 2 \, \mu \mu + \nu \pi 
- \tau \nu 
+ \Phi _4 \, . \label{Ricci110111}
\end{align}

\paragraph{Bianchi identity}
The Bianchi identity \eqref{eq-Bianchi-spinor} can be re-expressed as
\begin{multline*}
\varepsilon^{A C} \varepsilon^{B D}  \left( \partial_{A B} \Phi_{C D E F} - \, \partial_{E F} S \right) = \\
\varepsilon^{A C} \varepsilon^{B D}  \left( \gamma \ind{_{A B C}^{G}} \Phi \ind{_{D E F G}} + \gamma \ind{_{A B D}^{G}} \Phi \ind{_{E F C G}} + \gamma \ind{_{A B E}^{G}} \Phi \ind{_{F C D G}} + \gamma \ind{_{A B F}^{G}} \Phi \ind{_{C D E G}} \right) \, ,
\end{multline*}
so that taking components with respect to the spin dyad yields
\begin{align}
D \Phi _2 + \Delta \Phi _0 - 2 \, \delta \Phi _1 - D S & =  \left( 2 \, \pi - 4 \, \tau - 4 \, \alpha \right) \Phi _1 
- 2 \, \kappa \Phi _3 
+ \left( 4 \, \gamma - 2 \, \mu \right) \Phi _0 
+ 6 \, \rho \Phi _2 \, , \label{Bianchi00} \\
D \Phi _3 + \Delta \Phi _1 - 2 \, \delta \Phi _2 - \delta S & =
\left( 3 \, \pi - 3 \, \tau \right) \Phi _2 + \left( 4 \, \rho - 2 \, \epsilon \right) \Phi _3
- \kappa \Phi _4 + \nu \Phi _0 
+ \left( 2 \, \gamma 
- 4 \, \mu \right) \Phi _1 \, , \label{Bianchi01} \\
D \Phi _4 + \Delta \Phi _2 - 2 \, \delta \Phi _3 - \Delta S & =
\left( 4 \, \pi - 2 \, \tau + 4 \, \alpha \right) \Phi _3 
+ \left( 2 \, \rho - 4 \, \epsilon \right) \Phi _4 
+ 2 \, \nu \Phi _1 
- 6 \, \mu \Phi _2 \, . \label{Bianchi11}
\end{align}

\paragraph{Cotton tensor}
Finally, from the definition \eqref{eq-Cotton-spinor} of the Cotton tensor , one obtains
\begin{align*}
4 \, \varepsilon^{E F}   \partial_{(A| F|} \Phi_{B C D) E} & = A_{A B C D} +
\varepsilon^{E F}  \left( 12 \, \gamma \ind{_{(A| F| B}^{G}} \Phi \ind{_{C D) E G}} + 4 \, \gamma \ind{_{(A| F E}^{G}} \Phi \ind{_{|B C D) G}} \right)  \, ,
\end{align*}
with components
\begin{align}
4 \left( \delta \Phi _0 - D \Phi _1 \right) & = A _0 +  4 \left( 4 \,  \alpha - \pi \right) \Phi _0 
- 4 \left( 4 \, \rho + 2 \, \epsilon \right) \Phi _1 + 12 \, \kappa \Phi _2 \, , \label{Cotton0000} \\
\Delta \Phi _0 + 2 \, \delta \Phi _1 - 3 D \Phi _2 & = A _1 + \left( 2\, \mu + 4 \, \gamma \right) \Phi _0 
+ \left( 4 \, \alpha - 4 \, \tau - 6 \, \pi \right) \Phi _1
- 6 \, \rho \Phi _2 
+ 6 \, \kappa \Phi _3 \, , \label{Cotton0001} \\
2 \Delta \Phi _1 - 2 D \Phi _3 & = A _2 
+ 2 \,  \nu \Phi _0 
+ 2 \,  \kappa \Phi _4 
+ 4 \,  \gamma \Phi _1 
+ 4 \,  \epsilon \Phi _3 
-  6 \left(  \pi + \tau \right)  \Phi _2 \, , \label{Cotton0011} \\
 - D \Phi _4 - 2 \, \delta \Phi _3 + 3 \Delta \Phi _2 & = A _3 
+ \left( 2\, \rho + 4 \, \epsilon \right) \Phi _4 
+ \left( 4 \, \alpha - 4 \, \pi - 6 \, \tau \right) \Phi _3
- 6 \, \mu \Phi _2 
+ 6 \, \nu \Phi _1 \, , \label{Cotton0111} \\
 4 \left( - \delta \Phi _4 + \Delta \Phi _3 \right) & = A _4
 +  4 \left( 4 \,  \alpha - \tau \right) \Phi _4 
- 4 \left( 4 \, \mu + 2 \, \gamma \right) \Phi _3 + 12 \, \nu \Phi _2 \, . \label{Cotton1111}
\end{align}

\subsection{Reality conditions}\label{sec-real-cond}
There remain to impose suitable reality conditions on the null basis  $( k^a , \ell^a , n^a )$ so that the metric \eqref{eq-metric} has the correct signature. These are listed together with their effects on the spin coefficients and the components of the tracefree Ricci tensor and Cotton tensor.
\begin{itemize}
 \item Signature $(3,0)$: $\{ k^a , \ell^a , n^a \} \mapsto  \{ \overline{k^a} , \overline{\ell^a} , \overline{n^a} \} = \{ \ell^a , k^a , -n^a  \}$
\begin{align*}
\begin{pmatrix}
\kappa & \rho & \tau \\
\epsilon & \alpha & \gamma \\
\pi & \mu & \nu 
\end{pmatrix} & \mapsto
\begin{pmatrix}
\bar{\kappa} & \bar{\rho} & \bar{\tau} \\
\bar{\epsilon} & \bar{\alpha} & \bar{\gamma} \\
\bar{\pi} & \bar{\mu} & \bar{\nu} 
\end{pmatrix} =
\begin{pmatrix}
\nu & - \mu & \pi \\
- \gamma & \alpha & - \epsilon \\
\tau & - \rho & \kappa
\end{pmatrix} \, , \\
\begin{pmatrix}
 \Phi _0 \\ \Phi _1 \\ \Phi _2 \\ \Phi _3 \\ \Phi _4
 \end{pmatrix} & \mapsto 
 \begin{pmatrix}
  \overline{\Phi _0} \\ \overline{\Phi _1} \\ \overline{\Phi _2} \\ \overline{\Phi _3} \\ \overline{\Phi _4} 
  \end{pmatrix} = 
  \begin{pmatrix}
  \Phi _4 \\ - \Phi _3 \\ \Phi _2 \\ - \Phi _1 \\ \Phi _0 
  \end{pmatrix} \, ,
  & & &
\begin{pmatrix}
 A _0 \\ A _1 \\ A _2 \\ A _3 \\ A _4
 \end{pmatrix} & \mapsto 
 \begin{pmatrix}
  \overline{A _0} \\ \overline{A _1} \\ \overline{A _2} \\ \overline{A _3} \\ \overline{A _4} 
  \end{pmatrix} = 
  \begin{pmatrix}
  A _4 \\ - A _3 \\ A _2 \\ - A _1 \\ A _0 
  \end{pmatrix} \, .
\end{align*} 

  \item Signature $(2,1)$:
  \begin{itemize}
    \item[$\lozenge$] Real index $0$: $\{ k^a , \ell^a , n^a \} \mapsto  \{ \overline{k^a} , \overline{\ell^a} , \overline{n^a} \} = \{ \ell^a , k^a , n^a  \}$
\begin{align*}
\begin{pmatrix}
\kappa & \rho & \tau \\
\epsilon & \alpha & \gamma \\
\pi & \mu & \nu 
\end{pmatrix} & \mapsto
\begin{pmatrix}
\bar{\kappa} & \bar{\rho} & \bar{\tau} \\
\bar{\epsilon} & \bar{\alpha} & \bar{\gamma} \\
\bar{\pi} & \bar{\mu} & \bar{\nu} 
\end{pmatrix} =
\begin{pmatrix}
- \nu & - \mu & - \pi \\
- \gamma & - \alpha & - \epsilon \\
- \tau & - \rho & - \kappa
\end{pmatrix} \, ,
\\
\begin{pmatrix}
 \Phi _0 \\ \Phi _1 \\ \Phi _2 \\ \Phi _3 \\ \Phi _4 
 \end{pmatrix} & \mapsto 
 \begin{pmatrix}
 \overline{\Phi _0} \\ \overline{\Phi _1} \\ \overline{\Phi _2} \\ \overline{\Phi _3} \\ \overline{\Phi _4} 
 \end{pmatrix}
 = 
 \begin{pmatrix}
 \Phi _4 \\ \Phi _3 \\ \Phi _2 \\ \Phi _1 \\ \Phi _0 
 \end{pmatrix} \, ,
 & & &
 \begin{pmatrix}
 A _0 \\ A _1 \\ A _2 \\ A _3 \\ A _4 
 \end{pmatrix} & \mapsto 
 \begin{pmatrix}
 \overline{A _0} \\ \overline{A _1} \\ \overline{A _2} \\ \overline{A _3} \\ \overline{A _4} 
 \end{pmatrix}
 = 
 \begin{pmatrix}
 - A _4 \\ - A _3 \\ - A _2 \\ - A _1 \\ - A _0 
 \end{pmatrix} \, .
\end{align*}
  \item[$\lozenge$] Real index $1$: $\{ k^a , \ell^a , n^a \} \mapsto  \{ \overline{k^a} , \overline{\ell^a} , \overline{n^a} \} = \{ k^a , \ell^a , -n^a  \}$
\begin{align*}
\begin{pmatrix}
\kappa & \rho & \tau \\
\epsilon & \alpha & \gamma \\
\pi & \mu & \nu 
\end{pmatrix} & \mapsto
\begin{pmatrix}
\bar{\kappa} & \bar{\rho} & \bar{\tau} \\
\bar{\epsilon} & \bar{\alpha} & \bar{\gamma} \\
\bar{\pi} & \bar{\mu} & \bar{\nu} 
\end{pmatrix} =
\begin{pmatrix}
- \kappa & \rho & - \tau \\
\epsilon & - \alpha & \gamma \\
- \pi & \mu & - \nu
\end{pmatrix} \, ,
\\
\begin{pmatrix}
\Phi _0 \\ \Phi _1 \\ \Phi _2 \\ \Phi _3 \\ \Phi _4 
\end{pmatrix} & \mapsto 
\begin{pmatrix}
\overline{\Phi _0} \\ \overline{\Phi _1} \\ \overline{\Phi _2} \\ \overline{\Phi _3} \\ \overline{\Phi _4} 
\end{pmatrix} = 
\begin{pmatrix}
\Phi _0 \\ - \Phi _1 \\ \Phi _2 \\ - \Phi _3 \\ \Phi _4 
\end{pmatrix} \, ,
& & &
\begin{pmatrix}
A _0 \\ A _1 \\ A _2 \\ A _3 \\ A _4 
\end{pmatrix} & \mapsto 
\begin{pmatrix}
\overline{A _0} \\ \overline{A _1} \\ \overline{A _2} \\ \overline{A _3} \\ \overline{A _4} 
\end{pmatrix} = 
\begin{pmatrix}
- A _0 \\ A _1 \\ - A _2 \\ A _3 \\ - A _4 
\end{pmatrix} \, .
\end{align*}
  \end{itemize}
\end{itemize}

\section{A spinorial approach to the Goldberg-Sachs theorem}\label{sec-GS-spinor}
The aim of this appendix is to give alternative proofs for the results in the main text using the spinor calculus of Appendix \ref{app-spinor-calculus}. Throughout, $(\mcM,\bm{g})$ will denote an oriented three-dimensional (pseudo-)Riemannian manifold. Recall that there is a one-to-one correspondence between projective spinor fields and null complex line distributions, i.e. null structures, on $\mcM$. Some of the following results are already given and generalised to arbitrary dimensions in \cites{Taghavi-Chabert2012a,Taghavi-Chabert2013}.

The following proposition is a spinorial version of Proposition \ref{prop-intrinsic-torsion}, and relates the integrability properties of a null structure to differential conditions on its associated spinor field.
\begin{prop}\label{prop-intrinsic-torsion-spinor}
Let $\xi^A$ be a spinor field on $(\mcM,\bm{g})$ with associated null structure $\mcN$. Then
\begin{align*}
\mbox{$\mcN$ is co-integrable} & & \Longleftrightarrow & & \xi^A \xi^B \xi^C \nabla _{AB} \xi_C & = 0 \, , \\
\mbox{$\mcN$ is co-geodetic} & & \Longleftrightarrow & & \xi^B \xi^C \nabla _{AB} \xi_C & = 0 \, , \\
\mbox{$\mcN$ is parallel} & & \Longleftrightarrow & & \xi^C \nabla _{AB} \xi_C & = 0 \, .
\end{align*}
\end{prop}

\begin{proof}
The above result is already given in \cite{Taghavi-Chabert2013}. We can however use the NP formalism of appendix \ref{sec-NP} by taking $o^A := \xi^A$. Then
\begin{align*}
o^C \nabla_{AB} o_C & = \kappa \, \iota_A \iota_B - 2 \, \rho \, o_{(A} \iota_{B)} + \tau \, o_A o_B \, .
\end{align*}
Contracting the free indices with instances of $o^A$ give conditions on $\kappa$, $\rho$ and $\tau$, which, with reference to the the expression for $\left( \nabla_a k_{[b} \right) k_{c]}$, yields the result.
\end{proof}

Having translated the Lie bracket conditions of a null line and its orthogonal complement into spinorial differential equations, we can now re-express some of the results of sections  \ref{sec-alg-class}, \ref{sec-curvature-cond} and \ref{sec-GS}. In particular, with reference to Remark \ref{rem-spinor-Petrov}, the Petrov classification of the Weyl tensor can be expressed in the following terms.

\begin{lem}
Let $\xi^A$ be a pure spinor field on $(\mcM,\bm{g})$ with associated null structure $\mcN$. Then
\begin{itemize}
\item $\mcN$ is a principal null structure if and only if $\Phi_{ABCD} \xi^A \xi^B \xi^C \xi^D = 0$;
\item $\Phi_{ab}$ is algebraically special, i.e. of Petrov type II with $\mcN$ as multiple principal null structure if and only if $\Phi_{ABCD} \xi^A \xi^B \xi^C = 0$;
\item $\Phi_{ab}$ is of Petrov type III with $\mcN$ as multiple principal null structure if and only if $\Phi_{ABCD} \xi^A \xi^B = 0$;
\item $\Phi_{ab}$ is of Petrov type N with $\mcN$ as multiple principal null structure if and only if $\Phi_{ABCD} \xi^A = 0$.
\end{itemize}
\end{lem}

With this lemma, it is easy to compare the remaining results with those of sections \ref{sec-curvature-cond} and \ref{sec-GS}.

\begin{prop}[Integrability condition]\label{prop-int-cond-spinor}
Let $\xi^A$ be a spinor field on $(\mcM,\bm{g})$, and suppose it satisfies
\begin{align}\label{eq-strongly_foliating}
\xi^B \xi^C \nabla _{AB} \xi_C & = 0 \, .
\end{align}
Then
\begin{align}\label{eq-int-cond-spinor}
\Phi _{ABCD} \xi^A \xi^B \xi^C \xi^D & = 0 \, .
\end{align}
\end{prop}

Our first aim is to reformulate the obstruction to the existence of a co-geodetic multiple principal null structure of the tracefree Ricci tensor. Proposition \ref{prop-obstruction-GS} gave conditions on the components of the Cotton tensor $A_{abc}$ and the derivative of the Ricci scalar $R$. It turns out that the spinorial formalism gives very concise expressions for conditions \eqref{eq-obstruction-II}, \eqref{eq-obstruction-III} and \eqref{eq-obstruction-N}. Indeed using equation \eqref{eq-Bianchi+Cotton}, we obtain
\begin{align*}
4\,  \xi^B \xi^C \xi^D  \nabla \ind{_A^E} \Phi _{BCDE} & =  A _{ABCD} \xi^B \xi^C \xi^D  + 3 \, \xi_A \xi^B \xi^C \nabla \ind{_{BC}} S \, , \\
4 \, \xi^C \xi^D \nabla \ind{_A^E} \Phi _{BCDE} & =  A _{ABCD} \xi^C \xi^D  + 2 \, \xi \ind{_A} \xi^C \nabla \ind{_{BC}}  S - \varepsilon \ind{_{AB}} \xi^C \xi^D \nabla \ind{_{CD}} S  \, , \\
4 \, \xi^D \nabla \ind{_A^E} \Phi _{BCDE} & = A _{ABCD} \xi^D  + \xi \ind{_A} \nabla \ind{_{BC}} S - 2 \, \varepsilon \ind{_{A(B}} \xi^D \nabla \ind{_{C)D}} S \, .
\end{align*}

We can now re-express Proposition \ref{prop-obstruction-GS} as
\begin{prop}\label{prop-obstruction}
Let $\xi^A$ be a spinor field on $(\mcM,\bm{g})$. Suppose $\xi^A$ satisfies
\begin{align}\tag{\ref{eq-strongly_foliating}}
\xi^B \xi^C \nabla _{AB} \xi_C & = 0 \, .
\end{align}
Then,
\begin{align}
\xi^B \xi^C \xi^D \Phi _{ABCD} & = 0 & \Longrightarrow & & \xi^B \xi^C \xi^D  \nabla \ind{_A^E} \Phi _{BCDE} & = 0 \, , \label{eq-obstruction-II-sp} \\
\xi^C \xi^D \Phi _{ABCD} & = 0 & \Longrightarrow & & \xi^C \xi^D  \nabla \ind{_A^E} \Phi _{BCDE} & = 0 \, , \label{eq-obstruction-III-sp}\\
\xi^D \Phi _{ABCD} & = 0 & \Longrightarrow & & \xi^D  \nabla \ind{_A^E} \Phi _{BCDE} & = 0 \, . \label{eq-obstruction-N-sp}
\end{align}
\end{prop}

\begin{proof}
Assume $\xi^A$ satisfies \eqref{eq-strongly_foliating}. We first differentiate $\Phi _{ABCD} \xi^B \xi^C \xi^D = 0$ so that
\begin{align*}
0 = \nabla \ind{_A^E} \left( \Phi _{BCDE} \xi^B \xi^C \xi^D \right) & = \left( \nabla \ind{_A^E} \Phi _{BCDE} \right) \xi^B \xi^C \xi^D  + 3 \, \Phi _{BCDE} \left( \nabla \ind{_A^E} \xi^B \right) \xi^C \xi^D \, .
\end{align*}
The condition on $\Phi_{ABCD}$ can be rewritten as $\Phi _{ABCD} \xi^C \xi^D = \phi \, \xi_A \xi_B$ for some $\phi$. The second term then becomes $3 \, \phi \, \xi^B \xi^C \nabla \ind{_{AB}} \xi_C$, but this must vanish since $\xi^A$ satisfies \eqref{eq-strongly_foliating}. This proves \eqref{eq-obstruction-II-sp}.

The remaining cases are similar and left to the reader.
\end{proof}

\begin{rem}
Using the useful identity \eqref{eq-Bianchi+Cotton}, it is straightforward to see that the conditions on the RHS of \eqref{eq-obstruction-II-sp}, \eqref{eq-obstruction-III-sp} and \eqref{eq-obstruction-N-sp} are equivalent to the tensorial expression \eqref{eq-obstruction-II}, \eqref{eq-obstruction-III} and \eqref{eq-obstruction-N}.
\end{rem}

For conciseness, we combine the statements of Theorems \ref{thm-GS-II-gen}, \ref{thm-GS-III-gen} and \ref{thm-GS-N-gen}  into a single theorem.
\begin{thm}
Let $\xi^A$ be a spinor field on $(\mcM,\bm{g})$. Suppose $\xi^A$ satisfies any of the following conditions
\begin{enumerate}
\item $\xi^B \xi^C \xi^D \Phi _{ABCD} = 0$, $\xi^C \xi^D \Phi _{ABCD} \neq 0$, and $\xi^B \xi^C \xi^D  \nabla \ind{_A^E} \Phi _{BCDE}  = 0$; \label{cond1}
\item $\xi^C \xi^D \Phi _{ABCD} = 0$, $\xi^D \Phi _{ABCD} \neq 0$ and $\xi^C \xi^D  \nabla \ind{_A^E} \Phi _{BCDE}  = 0$;
\item $\xi^D \Phi _{ABCD} = 0$, $\Phi _{ABCD} \neq 0$ and $\xi^D  \nabla \ind{_A^E} \Phi _{BCDE}  = 0$.
\end{enumerate}
Then $\xi^A$ satisfies
\begin{align}\tag{\ref{eq-strongly_foliating}}
\xi^B \xi^C \nabla _{AB} \xi_C & = 0 \, .
\end{align}
\end{thm}

\begin{proof}
We assume the conditions given in case \ref{cond1}. We first differentiate $\Phi _{ABCD} \xi^B \xi^C \xi^D = 0$ so that
\begin{align*}
0 = \nabla \ind{_A^E} \left( \Phi _{BCDE} \xi^B \xi^C \xi^D \right) & = \left( \nabla \ind{_A^E} \Phi _{BCDE} \right) \xi^B \xi^C \xi^D  + 3 \, \Phi _{BCDE} \left( \nabla \ind{_A^E} \xi^B \right) \xi^C \xi^D \, .
\end{align*}
Since $\Phi_{ABCD}$ does not degenerate further with respect to $\xi^A$, we can write $\Phi _{ABCD} \xi^C \xi^D = \phi \, \xi_A \xi_B$ for some non-vanishing $\phi$. Hence,
\begin{align*}
0 & =  \xi^B \xi^C \xi^D  \nabla \ind{_A^E} \Phi _{BCDE}  + 3 \, \phi \, \xi^B \xi^C \nabla \ind{_{AB}} \xi_C \, .
\end{align*}
By assumption, the first term vanishes, and since $\phi$ is non-vanishing, we conclude $\xi^B \xi^C  \nabla \ind{_{AB}} \xi_C =0$.

We omit the proofs of the remaining cases, which are similar.
\end{proof}

Finally, Theorem \ref{thm-GS-hard} reads
\begin{thm}
Let $\xi^A$ be a spinor field on $(\mcM,\bm{g})$. Suppose $\xi^A$ satisfies
\begin{align}\tag{\ref{eq-strongly_foliating}}
\xi^B \xi^C \nabla _{AB} \xi_C & = 0 \, ,
\end{align}
and
\begin{align}\label{eq-diff-Ricci-cond}
\xi^B \xi^C \xi^D  \nabla \ind{_A^E} \Phi _{BCDE} & = 0 \, .
\end{align}
Then the tracefree Ricci tensor is algebraically special, i.e.
\begin{align*}
\Phi _{ABCD} \xi^B \xi^C \xi^D & = 0 \, ,
\end{align*}
\end{thm}

\begin{proof}
Assume $\xi^A$ satisfies \eqref{eq-strongly_foliating}. Then
\begin{itemize}
\item we can write
\begin{align}\label{eq-conseq-st-fol}
\xi^E \nabla_{AE} \xi_B & = \eta_A \xi_B \, , &
\xi^C \nabla_{AB} \xi_C & = \lambda \, \xi_A \xi_B \, .
\end{align}
for some $\eta_A$ and $\lambda$. Using the identity $\xi_C \nabla_{AB} \xi^C - \xi_A \nabla _{CB} \xi^C = - \xi^C \nabla_{CB} \xi_A$ tells us that
\begin{align}\label{eq-conseq-st-fol2}
\lambda \, \xi_A = \eta_A - \nabla_{AB} \xi^B \, .
\end{align}
\item by Proposition \ref{prop-int-cond-spinor}, the tracefree Ricci tensor satisfies
\begin{align}\label{eq-not-alg-sp}
\phi \, \xi_A & = \Phi _{ABCD} \xi^B \xi^C \xi^D 
\end{align}
for some function $\phi$.
\end{itemize}
Take the covariant derivative of \eqref{eq-not-alg-sp} and use the Leibnitz rule to get
\begin{align}\label{eq-diff-not-alg-sp}
\xi_E \nabla \ind{_A^E} \phi + \phi \nabla \ind{_A^E} \xi_E & = 3 \, \Phi _{BCDE} \left( \nabla \ind{_A^E} \xi^B \right) \xi^C \xi^D \, ,
\end{align}
where we have made use of the curvature assumption \eqref{eq-diff-Ricci-cond}.

We shall now suppose that $\phi$ does not vanish, and divide \eqref{eq-diff-not-alg-sp} through by $\phi$. Then
using \eqref{eq-strongly_foliating}, \eqref{eq-int-cond-spinor}, \eqref{eq-conseq-st-fol} and \eqref{eq-conseq-st-fol2} yields
\begin{align}\label{eq-potential-eq}
\xi^B \nabla_{AB} \ln \phi = 6 \eta_A - 4 \nabla_{AB} \xi^B =: \alpha_A \, .
\end{align}
The consistency condition for \eqref{eq-potential-eq} to be locally integrable can be obtained by applying $\xi^B \nabla_{AB}$ to \eqref{eq-potential-eq} and commuting the derivatives: we find
\begin{align}\label{eq-potential-cond}
\eta^A \alpha_A & = \xi^B \nabla \ind{_B^A} \alpha_A \, .
\end{align}
We proceed by checking that \eqref{eq-potential-cond} is indeed satisfied. Plugging the definition of $\alpha_A$ in the RHS of \eqref{eq-potential-eq} into \eqref{eq-potential-cond} yields
\begin{align}\label{eq-potential-cond2}
- 6 \, \xi^A \nabla_{AB} \eta^B + 4 \, \xi^A \nabla_{AB} \nabla \ind{^B_C} \xi^C & = - 4 \, \eta^A \nabla_{AB} \xi^B \, .
\end{align}
By commuting the covariant derivatives, the second term on the LHS of \eqref{eq-potential-cond2} becomes
\begin{align*}
\xi^A \nabla_{AB} \nabla \ind{^B_C} \xi^C & = - \xi^A \nabla_{CB} \nabla \ind{^B_A} \xi^C  + 2\, \xi^B \square_{BC} \xi^C \\
& = - \nabla_{CB} \left(  \xi^A \nabla \ind{^B_A} \xi^C  \right) + \left(   \nabla_{CB} \xi^A \right) \left( \nabla \ind{^B_A} \xi^C  \right) \, ,
\end{align*}
where we have made use of the fact that $\xi^B \square_{BC} \xi^C = 0$ in the first line, and the Leibnitz rule in the second line. The last term in the second line vanishes by symmetry consideration. Hence, using the definition of $\eta^A$ in \eqref{eq-conseq-st-fol}, we are left with $\xi^A \nabla_{AB} \nabla \ind{^B_C} \xi^C = - \nabla_{CB} \left(  \eta^B \xi^C  \right)$, which on substitution into \eqref{eq-potential-cond2} leads to $\xi^B \nabla_{BA} \eta^A = 0$. But now observe that
\begin{align*}
0 & = \xi_A \xi^B \nabla_{BC} \eta^C = \xi^B \nabla_{BC} \left( \eta^C \xi_A \right) = \xi^B \nabla_{BC} \left( \xi^D \nabla \ind{_D^C} \xi_A \right) = \xi^B \xi^D \nabla_{BC} \nabla \ind{_D^C} \xi_A = \Phi \ind{_{ABCD}} \xi^B \xi^C \xi^D \, ,
\end{align*}
which shows that $\Phi_{ABCD}$ is algebraically special in contradiction to our assumption that $\phi$ is non-vanishing. Hence the result.
\end{proof}

\bibliography{biblio}

\begin{bibdiv}
\begin{biblist}

\bib{Apostolov1997}{article}{
      author={Apostolov, V.},
      author={Gauduchon, P.},
       title={The {R}iemannian {G}oldberg-{S}achs theorem},
        date={1997},
     journal={Internat. J. Math.},
      volume={8},
      number={4},
       pages={421\ndash 439},
}

\bib{Baird1988}{article}{
      author={Baird, P.},
      author={Wood, J.~C.},
       title={Bernstein theorems for harmonic morphisms from {${\bf R}^3$} and
  {$S^3$}},
        date={1988},
        ISSN={0025-5831},
     journal={Math. Ann.},
      volume={280},
      number={4},
       pages={579\ndash 603},
         url={http://dx.doi.org/10.1007/BF01450078},
      review={\MR{939920 (90e:58027)}},
}

\bib{Baird1995}{incollection}{
      author={Baird, P.},
      author={Wood, J.~C.},
       title={{Monopoles, harmonic morphisms and spinor fields}},
        date={1995},
   booktitle={Further advances in twistor theory. {V}ol. {II} {L. J. Mason, L.
  P. Hughston, P. Z. Kobak eds (Harlow: Longman Scientific \& Technical)}},
      editor={L.~J.~Mason, P. Z.~Kobak, L. P.~Hughston},
   publisher={(Harlow: Longman Scientific \& Technical)},
       pages={49 \ndash  61},
}

\bib{Calvaruso2011}{article}{
      author={Calvaruso, G.},
       title={Contact {L}orentzian manifolds},
        date={2011},
        ISSN={0926-2245},
     journal={Differential Geom. Appl.},
      volume={29},
      number={suppl. 1},
       pages={S41\ndash S51},
         url={http://dx.doi.org/10.1016/j.difgeo.2011.04.006},
      review={\MR{2831998}},
}

\bib{Chinea1990}{article}{
      author={Chinea, D.},
      author={Gonzalez, C.},
       title={A classification of almost contact metric manifolds},
        date={1990},
        ISSN={0003-4622},
     journal={Ann. Mat. Pura Appl. (4)},
      volume={156},
       pages={15\ndash 36},
         url={http://dx.doi.org/10.1007/BF01766972},
      review={\MR{1080209 (91i:53047)}},
}

\bib{Coley2004}{article}{
      author={Coley, A.},
      author={Milson, R.},
      author={Pravda, V.},
      author={Pravdov{\'a}, A.},
       title={Classification of the {W}eyl tensor in higher dimensions},
        date={2004},
     journal={Classical Quantum Gravity},
      volume={21},
      number={7},
       pages={L35\ndash L41},
      eprint={arXiv:gr-qc/0401008},
}

\bib{Chow2010a}{article}{
      author={Chow, D. D.~K.},
      author={Pope, C.~N.},
      author={Sezgin, E.},
       title={Classification of solutions in topologically massive gravity},
        date={2010},
        ISSN={0264-9381},
     journal={Classical Quantum Gravity},
      volume={27},
      number={10},
       pages={105001, 26},
         url={http://dx.doi.org/10.1088/0264-9381/27/10/105001},
      review={\MR{2639100 (2011d:83098)}},
}

\bib{Chow2010}{article}{
      author={Chow, D. D.~K.},
      author={Pope, C.~N.},
      author={Sezgin, E.},
       title={Kundt spacetimes as solutions of topologically massive gravity},
        date={2010},
        ISSN={0264-9381},
     journal={Classical Quantum Gravity},
      volume={27},
      number={10},
       pages={105002, 19},
         url={http://dx.doi.org/10.1088/0264-9381/27/10/105002},
      review={\MR{2639101 (2011b:83022)}},
}

\bib{Deser1982}{article}{
      author={Deser, S.},
      author={Jackiw, R.},
      author={Templeton, S.},
       title={Topologically massive gauge theories},
        date={1982},
        ISSN={0003-4916},
     journal={Ann. Physics},
      volume={140},
      number={2},
       pages={372\ndash 411},
         url={http://dx.doi.org/10.1016/0003-4916(82)90164-6},
      review={\MR{665601 (84j:81128)}},
}

\bib{Durkee2009}{article}{
      author={Durkee, M.},
      author={Reall, H.~S.},
       title={A higher dimensional generalization of the geodesic part of the
  {G}oldberg-{S}achs theorem},
        date={2009},
        ISSN={0264-9381},
     journal={Classical Quantum Gravity},
      volume={26},
      number={24},
       pages={245005, 14},
         url={http://dx.doi.org/10.1088/0264-9381/26/24/245005},
      review={\MR{2570838}},
}

\bib{Gover2011}{article}{
      author={Gover, A.~R.},
      author={Hill, C.~D.},
      author={Nurowski, P.},
       title={Sharp version of the {G}oldberg-{S}achs theorem},
        date={2011},
        ISSN={0373-3114},
     journal={Ann. Mat. Pura Appl. (4)},
      volume={190},
      number={2},
       pages={295\ndash 340},
         url={http://dx.doi.org/10.1007/s10231-010-0151-4},
      review={\MR{2786175 (2012j:53057)}},
}

\bib{Goldberg2009}{article}{
      author={Goldberg, J.},
      author={Sachs, R.},
       title={{Republication of: A theorem on Petrov types}},
        date={2009},
        ISSN={0001-7701},
     journal={General Relativity and Gravitation},
      volume={41},
       pages={433\ndash 444},
         url={http://dx.doi.org/10.1007/s10714-008-0722-5},
        note={10.1007/s10714-008-0722-5},
}

\bib{Hill2008}{article}{
      author={Hill, C.~D.},
      author={Lewandowski, J.},
      author={Nurowski, P.},
       title={Einstein's equations and the embedding of 3-dimensional {CR}
  manifolds},
        date={2008},
        ISSN={0022-2518},
     journal={Indiana Univ. Math. J.},
      volume={57},
      number={7},
       pages={3131\ndash 3176},
         url={http://dx.doi.org/10.1512/iumj.2008.57.3473},
      review={\MR{2492229 (2010e:32035)}},
}

\bib{Hughston1988}{article}{
      author={Hughston, L.~P.},
      author={Mason, L.~J.},
       title={A generalised {K}err-{R}obinson theorem},
        date={1988},
     journal={Classical Quantum Gravity},
      volume={5},
      number={2},
       pages={275\ndash 285},
}

\bib{Hill2009}{article}{
      author={Hill, C.~D.},
      author={Nurowski, Pawe{\l}},
       title={Intrinsic geometry of oriented congruences in three dimensions},
        date={2009},
        ISSN={0393-0440},
     journal={J. Geom. Phys.},
      volume={59},
      number={2},
       pages={133\ndash 172},
         url={http://dx.doi.org/10.1016/j.geomphys.2008.10.001},
      review={\MR{2492187 (2009k:53058)}},
}

\bib{Jeffryes1995}{incollection}{
      author={Jeffryes, B.~P.},
       title={A six-dimensional `{P}enrose diagram'},
        date={1995},
   booktitle={Further advances in twistor theory. {V}ol. {II} {L. J. Mason, L.
  P. Hughston, P. Z. Kobak eds (Harlow: Longman Scientific \& Technical)}},
   publisher={Longman Scientific \& Technical},
       pages={85\ndash 87},
}

\bib{Kopczy'nski1997}{article}{
      author={Kopczy{\'n}ski, W.},
       title={Pure spinors in odd dimensions},
        date={1997},
     journal={Classical Quantum Gravity},
      volume={14},
      number={1A},
       pages={A227\ndash A236},
        note={Geometry and physics},
}

\bib{Kundt1962}{article}{
      author={Kundt, W.},
      author={Thompson, A.},
       title={Le tenseur de {W}eyl et une congruence associ\'ee de
  g\'eod\'esiques isotropes sans distorsion},
        date={1962},
     journal={C. R. Acad. Sci. Paris},
      volume={254},
       pages={4257\ndash 4259},
      review={\MR{0156677 (27 \#6597)}},
}

\bib{Kopczy'nski1992}{article}{
      author={Kopczy{\'n}ski, W.},
      author={Trautman, A.},
       title={Simple spinors and real structures},
        date={1992},
     journal={J. Math. Phys.},
      volume={33},
      number={2},
       pages={550\ndash 559},
}

\bib{Mason2010}{article}{
      author={{Mason}, L.},
      author={{Taghavi-Chabert}, A.},
       title={{Killing-Yano tensors and multi-Hermitian structures}},
        date={2010-06},
     journal={Journal of Geometry and Physics},
      volume={60},
       pages={907\ndash 923},
      eprint={arXiv:0805.3756},
}

\bib{Milson2013}{article}{
      author={Milson, R.},
      author={Wylleman, L.},
       title={Three-dimensional spacetimes of maximal order},
        date={2013},
        ISSN={0264-9381},
     journal={Classical Quantum Gravity},
      volume={30},
      number={9},
       pages={095004, 25},
         url={http://dx.doi.org/10.1088/0264-9381/30/9/095004},
      review={\MR{3046469}},
}

\bib{Nurowski2002}{article}{
      author={Nurowski, P.},
      author={Trautman, A.},
       title={Robinson manifolds as the {L}orentzian analogs of {H}ermite
  manifolds},
        date={2002},
     journal={Differential Geom. Appl.},
      volume={17},
      number={2-3},
       pages={175\ndash 195},
        note={8th International Conference on Differential Geometry and its
  Applications (Opava, 2001)},
}

\bib{Nurowski1993}{thesis}{
      author={Nurowski, P.},
       title={Einstein equations and {C}auchy-{R}iemann geometry},
        type={Ph.D. Thesis},
        date={1993},
}

\bib{Nurowski1996}{article}{
      author={Nurowski, P.},
       title={Optical geometries and related structures},
        date={1996},
        ISSN={0393-0440},
     journal={J. Geom. Phys.},
      volume={18},
      number={4},
       pages={335\ndash 348},
         url={http://dx.doi.org/10.1016/0393-0440(95)00012-7},
      review={\MR{1383221 (97g:83083)}},
}

\bib{Ortaggio2013a}{article}{
      author={Ortaggio, M.},
      author={Pravda, V.},
      author={Pravdov\'{a}, A.},
       title={{On the Goldberg-Sachs theorem in higher dimensions in the
  non-twisting case}},
        date={2013},
     journal={Class.Quant.Grav.},
      volume={30},
       pages={075016},
      eprint={1211.2660},
}

\bib{Ortaggio2012a}{article}{
      author={Ortaggio, M.},
      author={Pravda, V.},
      author={Pravdov\'{a}, A.},
      author={Reall, H.~S.},
       title={{On a five-dimensional version of the Goldberg-Sachs theorem}},
        date={2012},
     journal={Class.Quant.Grav.},
      volume={29},
       pages={205002},
      eprint={1205.1119},
}

\bib{Przanowski1983}{article}{
      author={Przanowski, M.},
      author={Broda, B.},
       title={Locally {K}\"ahler gravitational instantons},
        date={1983},
        ISSN={0587-4254},
     journal={Acta Phys. Polon. B},
      volume={14},
      number={9},
       pages={637\ndash 661},
      review={\MR{732971 (85g:83016b)}},
}

\bib{Penrose1960}{article}{
      author={Penrose, R.},
       title={A spinor approach to general relativity},
        date={1960},
        ISSN={0003-4916},
     journal={Ann. Physics},
      volume={10},
       pages={171\ndash 201},
      review={\MR{0115765 (22 \#6563)}},
}

\bib{Petrov2000}{article}{
      author={Petrov, A.~Z.},
       title={The classification of spaces defining gravitational fields},
        date={2000},
        ISSN={0001-7701},
     journal={Gen. Relativity Gravitation},
      volume={32},
      number={8},
       pages={1665\ndash 1685},
         url={http://dx.doi.org/10.1023/A:1001910908054},
        note={Reprint of Kazan. Gos. Univ. U{\v{c}}. Zap. {{\bf{1}}14} (1954),
  no. 8, 55--69 [ MR0076401 (17,892g)]},
      review={\MR{1784371 (2002h:83028)}},
}

\bib{Penrose1984}{book}{
      author={Penrose, R.},
      author={Rindler, W.},
       title={Spinors and space-time. {V}ol. 1},
      series={Cambridge Monographs on Mathematical Physics},
   publisher={Cambridge University Press},
     address={Cambridge},
        date={1984},
        note={Two-spinor calculus and relativistic fields},
}

\bib{Penrose1986}{book}{
      author={Penrose, R.},
      author={Rindler, W.},
       title={Spinors and space-time. {V}ol. 2},
      series={Cambridge Monographs on Mathematical Physics},
   publisher={Cambridge University Press},
     address={Cambridge},
        date={1986},
        note={Spinor and twistor methods in space-time geometry},
}

\bib{Robinson1963}{article}{
      author={Robinson, I.},
      author={Schild, A.},
       title={{Generalization of a Theorem by Goldberg and Sachs}},
        date={1963},
     journal={Journal of Mathematical Physics},
      volume={4},
      number={4},
       pages={484\ndash 489},
         url={http://link.aip.org/link/?JMP/4/484/1},
}

\bib{Taghavi-Chabert2011}{article}{
      author={Taghavi-Chabert, A.},
       title={{Optical structures, algebraically special spacetimes, and the
  Goldberg-Sachs theorem in five dimensions}},
        date={2011},
     journal={Class. Quant. Grav.},
      volume={28},
       pages={145010 (32pp)},
      eprint={1011.6168},
}

\bib{Taghavi-Chabert2012}{article}{
      author={Taghavi-Chabert, A.},
       title={The complex {G}oldberg-{S}achs theorem in higher dimensions},
        date={2012},
        ISSN={0393-0440},
     journal={J. Geom. Phys.},
      volume={62},
      number={5},
       pages={981\ndash 1012},
         url={http://dx.doi.org/10.1016/j.geomphys.2012.01.012},
      review={\MR{2901842}},
}

\bib{Taghavi-Chabert2012a}{article}{
      author={Taghavi-Chabert, A.},
       title={{Pure spinors, intrinsic torsion and curvature in even
  dimensions}},
        date={2012},
      eprint={1212.3595},
}

\bib{Taghavi-Chabert2013}{article}{
      author={Taghavi-Chabert, A.},
       title={{Pure spinors, intrinsic torsion and curvature in odd
  dimensions}},
        date={2013},
      eprint={1304.1076},
}

\bib{Taghavi-Chabert2014}{article}{
      author={Taghavi-Chabert, A.},
       title={{The curvature of almost Robinson manifolds}},
        date={2014},
      eprint={1404.5810},
}

\bib{Witten1959}{article}{
      author={Witten, L.},
       title={Invariants of general relativity and the classification of
  spaces},
        date={1959},
     journal={Phys. Rev. (2)},
      volume={113},
       pages={357\ndash 362},
      review={\MR{0115764 (22 \#6562)}},
}

\end{biblist}
\end{bibdiv}
\bibliographystyle{ieeetr}

\end{document}